\documentclass[[11pt]{article}

\usepackage{amsmath,amssymb,latexsym,cite,amsthm}
\usepackage{epsf, pstricks,pgf,xcolor}
\usepackage{epsfig}
\usepackage{subdepth}
\usepackage{bbm}
\usepackage{dsfont}
\usepackage{stmaryrd}
\usepackage{scrextend}
\usepackage{enumerate}
\usepackage{algpseudocode}
\usepackage{algorithm}
\usepackage{framed}
\usepackage{mdframed}
\usepackage{graphicx}
\usepackage{wrapfig}
\usepackage{tikz}

\usepackage[left=1 in,top=1 in,right=1 in,bottom=1 in]{geometry}

\newtheorem{thm}{Theorem}[section]
\newtheorem{theorem}[thm]{Theorem}
\newtheorem{corollary}[thm]{Corollary}
\newtheorem{lemma}[thm]{Lemma}

\newtheorem{definition}[thm]{Definition}

\newtheorem{claim}[thm]{Claim}

\def\cS{{\mathcal{S}}}
\def\cV{{\mathcal{V}}}

\def\cS{{\mathcal{S}}}

\graphicspath{{../figs/}{figs/}}

\begin{document}


\title{Information-Theoretic Lower Bounds on the Storage Cost of Shared Memory Emulation}

\author{Viveck R. Cadambe\\
	EE Department,\\
	Pennsylvania State University, \\
	University Park, PA, USA\\
	{\tt viveck@engr.psu.edu}
	\and
Zhiying Wang\\
CPCC Center\\
University of California, Irvine\\
{Irvine, CA, USA}\\
	{\tt zhiying@uci.edu}
\and
Nancy Lynch\\
Department of EECS\\
MIT \\
{Cambridge, MA, USA}\\
	{\tt lynch@csail.mit.edu}}
\date{}
\maketitle



\begin{abstract}
The focus of this paper is to understand storage costs of emulating an atomic shared memory over an asynchronous, distributed message passing system. Previous literature has developed several shared memory emulation algorithms based on replication and erasure coding techniques, and analyzed the storage costs of the proposed algorithms. In this paper, we present information-theoretic lower bounds on the storage costs incurred by shared memory emulation algorithms. Our storage cost lower bounds are universally applicable, that is, we make no assumption on the structure of the algorithm or the method of encoding the data. 

We consider an arbitrary algorithm $A$ that implements an atomic multi-writer single-reader (MWSR) shared memory variable whose values come from a finite set $\mathcal{V}$ over a system of $N$ servers connected by point-to-point asynchronous links. We require that in every fair execution of algorithm $A$ where the number of server failures is smaller than a parameter $f$, every operation invoked at a non-failing client terminates. We define the storage cost of a server in algorithm $A$ as the logarithm (to base 2) of number of states it can take on; the total-storage cost of algorithm $A$ is the  sum of the storage cost of all servers. We develop three lower bounds on the storage cost of algorithm $A$.

\begin{itemize}
\item In our first lower bound, we show that if algorithm $A$ does not use server gossip, then the total storage cost is lower bounded by $2 \frac{N}{N-f+1}\log_2|\mathcal{V}|-o(\log_2|\mathcal{V}|)$. 
\item In our second lower bound we show that the total storage cost is at least $2 \frac{N}{N-f+2} \log_{2}|\mathcal{V}|-o(\log_{2}|\mathcal{V}|)$ even if the algorithm uses server gossip.

\item In our third lower bound, we consider algorithms where the write protocol sends information about the value in at most one phase. For such algorithms, we show that the total storage cost is at least $\nu^* \frac{N}{N-f+\nu^*-1} \log_2( |\mathcal{V}|)- o(\log_2(|\mathcal{V}|),$ where $\nu^*$ is the minimum of $f+1$ and the number of active write operations of an execution.
\end{itemize}

Our first and second lower bounds are approximately twice as strong as the previously known bound of $\frac{N}{N-f}{\log_{2}|\mathcal{V}|}$. Furthermore, our first two lower bounds apply even for regular, single-writer single-reader (SWSR) shared memory emulation algorithms. Our third lower bound is much larger than our first and second lower bounds, although it is applicable to a smaller class of algorithms where the write protocol has certain restrictions. In particular, our third bound is comparable to the storage cost achieved by most shared memory emulation algorithms in the literature, which naturally fall under the class of algorithms studied. Our proof ideas are inspired by recent results in coding theory.

\end{abstract}
\newpage



\section{Introduction}
The emulation of a consistent, fault-tolerant read-write shared memory in a distributed, asynchronous, storage system has been an active area of research in distributed computing theory. In their celebrated paper \cite{ABD}, Attiya, Bar-Noy, and Dolev devised a fault-tolerant algorithm for emulating a shared memory that achieves atomic consistency (linearizability) \cite{Lamport86,Herlihy90}. Consider a distributed system with server nodes, write client and read client nodes, all of which are connected by point-to-point asynchronous links. The ideas of \cite{ABD} can be used to design server, write and read protocols that implement an atomic shared memory even if the write and read operations are invoked concurrently with the following guarantee: every read or write operation invoked at a non-failing client terminates so long as the set of servers that fail is restricted to a minority. The algorithm of \cite{ABD} used a replication-based storage scheme at the servers to attain fault tolerance. Following \cite{ABD}, several papers \cite{HGR, DGL, AJX, CT, FAB, dobre_powerstore, Cadambe_Lynch_Medard_Musial_NCA, androulaki2014_separate_metadata} have developed algorithms that use \emph{erasure coding} instead of replication for fault tolerance, with the goal of improving upon the storage efficiency of \cite{ABD}. 

In erasure coding\footnote{Server failures are modeled as erasures of codeword symbols; hence the term, erasure coding.} which is studied in classical coding theory, each server stores a function of the value called a codeword symbol. A decoder that is able to access a sufficient number of codeword symbols recovers the value. The number of bits used to represent a codeword symbol is typically much smaller than the number of bits used to represent the value. As a consequence, erasure coding is well known to lead to smaller storage costs as compared to replication in the classical coding-theoretic set-up (See, for example, \cite{dimakis_networkcoding, Patterson_raid, Cassuto_tutorial}). Here, we aim to understand storage costs of shared memory emulation, where in contrast with the classical coding-theoretic setup, multiple versions of the data object are to be stored in a consistent manner.

When erasure coding is used for shared memory emulation, new challenges arise. Since, in erasure coding, each server stores a codeword symbol and not the entire value, a read operation has to obtain a sufficient number of codeword symbols to decode the value being stored. When a write operation begins to write a new version of the data object, the old version cannot be deleted from the servers until a sufficient number of codeword symbols corresponding to the new version have been propagated to the servers. As a consequence, servers have to store codeword symbols corresponding to multiple versions of the data object to ensure that a reader can decode an atomically consistent version. Previous erasure coding based shared memory emulation algorithms  \cite{HGR, DGL, AJX, CT, FAB, dobre_powerstore, Cadambe_Lynch_Medard_Musial_NCA} have noted that the number of versions to be stored at a server can be large if there are a large number of ongoing or failed write operations whose codeword symbols have not been propagated sufficiently. Because servers store codeword symbols corresponding to multiple versions, the storage cost of using erasure coding can be large, even if the number of bits in each codeword symbol is small compared to the number of bits used to represent the value. 

Despite the vast amount of literature in the study of storage costs of shared memory emulation, some compelling and fundamental questions remain unanswered. Since a server can store an arbitrary function of all the symbols it receives, can we develop a sophisticated storage strategy that somehow compresses multiple versions at the servers and thereby results in smaller storage costs? If we add multiple phases to read and write protocols or include other algorithmic novelties, can we reduce the storage cost of shared memory emulation? {In our paper, we obtain insights into these questions by developing novel impossibility results that lower bound the storage cost of an arbitrary  atomic shared memory emulation algorithm.} 
\vspace{-8pt}
\section{Summary of Results and Comparisons with Related Work}
\label{sec:comparisons}
\vspace{-5pt}
In this section, we first summarize the shared memory emulation and the classical coding theory set-ups. We then describe our storage cost lower bounds in Theorems \ref{thm:second} and \ref{thm:third}. Then, we describe 
our storage cost lower bound related to Theorem \ref{thm:fourth}. Finally we compare our results to previously derived storage cost lower bounds. 

\vspace{-5pt}
\subsection{Set up}
\vspace{-5pt}

\textbf{Shared Memory Emulation Set-up:}
We consider an arbitrary algorithm $A$ that implements, over a network of $N$ servers connected by point-to-point asynchronous links, an atomic multi-writer single-reader (MWSR) shared memory variable whose values come from a finite set $\mathcal{V}$. The algorithm $A$ is required to ensure that all operations terminate so long as the number of server failures is no larger than a parameter $f$. 
The storage cost of a server in algorithm $A$ is measured as the logarithm of the number of possible states of the server, and the storage cost of algorithm $A$ is the {total} storage cost over all the servers.

\textbf{Classical Coding Theory Set-up:}
Consider a system with $N$ servers, where a single version of a data object whose values come from a finite set $\mathcal{V}$ is to be stored. The value of the data object must be recoverable, so long as the number of server failures is no larger than a parameter $f$. The classical \emph{Singleton bound} \cite{LinCostello_Book, Roth_CodingTheory} in coding theory implies that the {total} storage cost is at least $\frac{N\log_{2}|\mathcal{V}|}{N-f}$ bits\footnote{For the sake of the discussion here, we assume that $|\mathcal{V}|$ is a power of $2$. We refer the reader to \cite{LinCostello_Book, Cadambe_Lynch_Medard_Musial_new} for more details about erasure coding.}. The lower bound of $\frac{N \log_{2}|\mathcal{V}|  }{N-f}$  on the storage cost is known to be tight in the classical coding-theoretic set-up for large values of $|\mathcal{V}|$ \cite{Roth_CodingTheory, LinCostello_Book}.

The power of erasure coding is transparent when we want to design a storage system that tolerates failures of $f$ server nodes and the number of server nodes $N$ can be chosen freely. If we use replication, every server stores $\log_{2} |\mathcal{V}|$ bits. Since we need at least $f+1$ servers to tolerate $f$ server failures, the {total} storage cost of the system is at least $(f+1) \log_{2}|\mathcal{V}|$ bits.  In contrast, if we use erasure coding, the total storage cost of the system is $\frac{N \log_{2}|\mathcal{V}|}{N-f},$  which approaches $\log_{2}|\mathcal{V}|$ as $N$ increases. If $N$ is sufficiently large, the storage cost of replication is approximately $f+1$ times the storage cost of erasure coding. 

\vspace{-5pt}
\subsection{Motivation and Summary - Theorems \ref{thm:second} and \ref{thm:third}}
\vspace{-5pt}

\textbf{Motivation:} The classical coding-theoretic model does not model clients or channels, and therefore differs significantly from the shared memory emulation model. However, the storage cost lower bound of $\frac{\log_2 |\mathcal{V}|}{N-f}$ described by the Singleton bound is, in fact, applicable in the context of shared memory emulation as well. We provide the first formal proof of the lower bound in Appendix \ref{sec:firstthm}; in particular, we show that for any SWSR regular shared memory emulation algorithm that implements a read write data object whose values come from a set $\mathcal{V}$, the {total} storage cost is at least $\frac{N \log_2 |\mathcal{V}|}{N-f}$. The natural lower bound of  $\frac{N \log_2 |\mathcal{V}|}{N-f}$ inspires the following question. 

\begin{mdframed}
\textbf{Question 1:} Does there exist an atomic shared memory emulation algorithm whose storage cost is equal to $\frac{N}{N-f} \log_{2}|\mathcal{V}|?$
\end{mdframed}
\textbf{Summary of Theorems \ref{thm:second} and \ref{thm:third}:}
In this paper, we answer the above question in the negative by proving storage cost lower bounds that are stronger than $\frac{N}{N-f} \log_{2}|\mathcal{V}|.$ In Theorems \ref{thm:second} and \ref{thm:third} 
we show that the {total}-storage cost of single-writer-single-reader (SWSR) \emph{regular} shared memory emulation algorithm is at least $\frac{2 N}{N-f+2} \log_2|\mathcal{V}|-o(\log_2|\mathcal{V}|).$ In particular, if $f$ is fixed and $N$ is chosen freely, the total-storage cost lower bound of Theorems \ref{thm:second} and \ref{thm:third} approach $\frac{2N}{N-f} \log_{2} |\mathcal{V}|-o(|\log_2|\mathcal{V}|)$ as $N$ increases; therefore the bounds of Theorems \ref{thm:second} and \ref{thm:third} are twice as large as the previously known lower bound.  Recall that regularity \cite{Lamport86} is a weaker consistency model as compared with atomicity. Since Theorems \ref{thm:second} and \ref{thm:third} apply for regular SWSR shared memory emulation algorithm, they automatically apply for atomic MWSR shared memory emulation algorithms. Theorem \ref{thm:second} describe a storage cost lower bound for algorithms which do not use server gossip. Theorem \ref{thm:third} describes a lower bound for \emph{any} shared memory emulation algorithm, including algorithms that use server gossip. Our storage cost lower bounds are universal in nature, that is, we make no assumption on the structure of the protocols or the method of data storage. Because we answer Question 1 in the negative, an important implication is that there is an unavoidable price, in terms of storage cost, to ensure regularity in a shared memory emulation system. We next discuss the tightness of Theorems \ref{thm:second} and \ref{thm:third} in the context of previously derived storage cost {upper} bounds.

\vspace{-5pt}
\subsection{Motivation and Summary - Theorem \ref{thm:fourth}}
In the sequel, we define the number of active write operations at point $P$ of an execution as the number of write operations which have begun before the point $P$ but not yet terminated or failed at point $P$. The number of active write operations of an execution is the supremum, over all points of the execution, of the number of active write operations at the points of the execution. 

\textbf{Motivation:} {There is a growing body of literature related to erasure coding based shared memory emulation algorithms \cite{HGR, DGL, AJX, CT, FAB, dobre_powerstore, Cadambe_Lynch_Medard_Musial_new,  Cadambe_Lynch_Medard_Musial_NCA, androulaki2014_separate_metadata, spiegelman2015space}. These algorithms differ in their structure, liveness conditions on operation termination, and their communication costs. A common insight that applies to all the algorithms of \cite{HGR, DGL, AJX, CT, FAB, dobre_powerstore, Cadambe_Lynch_Medard_Musial_new,  Cadambe_Lynch_Medard_Musial_NCA, androulaki2014_separate_metadata, spiegelman2015space} is that, among the class of all executions with at most $\nu$ active write operations, the worst case storage cost of implementing an atomic shared memory object whose values come from a finite set $\mathcal{V}$ is \emph{at least} $\nu \frac{N\log_{2}|\mathcal{V}|}{N-f}.$ In fact, references \cite{DGL, Cadambe_Lynch_Medard_Musial_NCA,androulaki2014_separate_metadata, CT} conduct a formal analysis of the incurred storage cost and show that the storage cost incurred\footnote{{There are subtle differences in the storage cost incurred by the algorithms of \cite{Cadambe_Lynch_Medard_Musial_new, Cadambe_Lynch_Medard_Musial_NCA,androulaki2014_separate_metadata, CT, DGL}. Nonetheless, $\nu \frac{N}{N-f}\log_{2}|\mathcal{V}|$ is a lower bound on the cost incurred by these algorithms in the \emph{worst case}, among executions where the number of active writes is bounded by $\nu$}.} is approximately $\nu \frac{N \log_{2}|\mathcal{V}|}{N-f}$.
While the prior works highlight the benefit of erasure coding when the number of active writes is small, the storage cost benefits of erasure coding vanish as the number of active writes increases. In particular, for a sufficiently large value of $\nu$, erasure coding based algorithms can even have a higher storage cost as compared to replication based algorithms \cite{ABD, LDR}, which incur a storage cost of $\Theta(f) \log_{2}|\mathcal{V}|$ irrespective of the number of active writes. 

In contrast with the storage cost upper bounds in literature, our lower bounds of Theorem \ref{thm:second} and \ref{thm:third} do not depend on the number of active writes. Furthermore, if $f$ is proportional to $N,$ then storage cost lower bounds of Theorem \ref{thm:second} and \ref{thm:third} are both $o(f) \log_{2}|\mathcal{V}|+o(\log_{2}|\mathcal{V}|).$ Prior literature in conjunction with our results of Theorems \ref{thm:second} and \ref{thm:third} motivates the following question:

\begin{mdframed}
	\textbf{Question 2:} 
	Can we develop an algorithm whose storage cost, when $f$ is proportional to $N$, is as small as $o(f)\log_{2}|\mathcal{V}|$ and does not grow with the number of active writes? 
\end{mdframed}
\textbf{Summary of Theorem \ref{thm:fourth}:} 
We provide partial answer to Question 2 in our lower bound presented in {Theorem \ref{thm:fourth}}. The lower bound states that the answer to Question 2 is negative, if the write protocol of the algorithm satisfies certain technical conditions described in Section \ref{sec:fourth}. Informally speaking, the technical conditions in Section \ref{sec:fourth} imply that the write operation is executed in phases, and a message containing information about the value is sent to the servers in at most one phase per write operation. For any atomic MWSR algorithm that ensures that all operations terminate in every execution where the active number of write operations is at most $\nu$ and the number of server failures is at most $f$, Theorem \ref{thm:fourth} shows that if the write protocol satisfies the conditions stated in Section \ref{sec:fourth}, then the storage cost cannot be smaller than  $\nu^{*} \frac{\log_2|\mathcal{V}|}{N-f+\nu^{*}-1}-o(\log|\mathcal{V}|),$ where $\nu^{*}$ is the minimum of $\{f+1,\nu\}.$

Theorem \ref{thm:fourth} is interesting from a conceptual viewpoint since it captures the dependence of the storage cost on the degree of concurrency that has been noticed in the upper bounds of \cite{DGL, Cadambe_Lynch_Medard_Musial_new, Cadambe_Lynch_Medard_Musial_NCA,androulaki2014_separate_metadata, CT}. In particular, the bound of Theorem \ref{thm:fourth} can be much larger than  the bounds of Theorems \ref{thm:second} and \ref{thm:third}, if the parameters $\nu$ and $f$ are sufficiently large. If the number of active write operations exceeds $(f+1)$, then our storage cost lower bound of Theorem \ref{thm:fourth}, which equals $(f+1)\log_2|\mathcal{V}|-o(\log_2|\mathcal{V}|),$ implies that replication based algorithms are approximately optimal in the class of algorithms described in the theorem.

The class of algorithms that satisfy the conditions stated in {Section \ref{sec:fourth}} include a majority of the algorithms in literature \cite{DGL, AJX, CT, FAB, dobre_powerstore, Cadambe_Lynch_Medard_Musial_NCA}. We refer the reader to Section \ref{sec:fourth} for a more detailed justification. {Theorem \ref{thm:fourth}}, in the stated form, does not apply to a few algorithms \cite{androulaki2014_separate_metadata, HGR} because these protocols send messages related to the value of the write operation in two phases; one phase is used to send a hash of the value for client verification purposes, and a second phase is used to send codeword symbols corresponding to the value. In related discussions in Section \ref{sec:fourth}, we conjecture that the statement of {Theorem \ref{thm:fourth}} and the proof can be modified, without deviating too much from our approach, to apply to a larger class of algorithms which include \cite{androulaki2014_separate_metadata, HGR}.}

In Figure \ref{fig_compare} we compare the proposed total-storage lower bounds with the previous achievable upper bounds.

\begin{figure}
\centering
\begin{picture}(320,240)
\put(0,0){\includegraphics[scale=0.8]{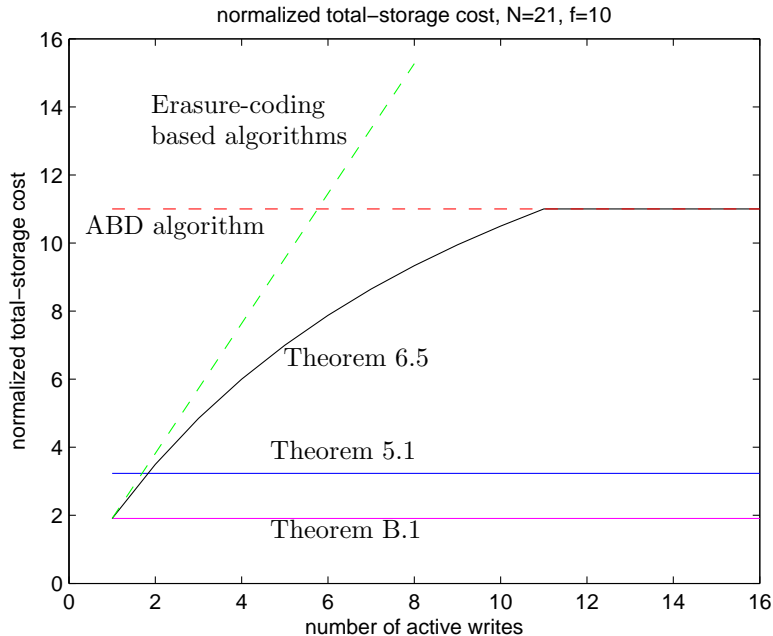}}
\put(120,45){Theorem \ref{thm:first}}
\put(120,75){Theorem \ref{thm:third}}
\put(125,110){Theorem \ref{thm:fourth}}
\put(50,160){ABD algorithm}
\put(75,200){\parbox{1.2in}{Erasure-coding based algorithms}}
\end{picture}
\caption{Storage cost upper and lower bounds for $N=21$ servers and $f=10$ server failures. We plot the total-storage cost normalized by $\log_2|\mathcal{V}|$ when $|\mathcal{V}| \to \infty$. Theorems \ref{thm:third}, \ref{thm:fourth}, \ref{thm:first} are lower bounds obtained in this paper, that corresponds to $2\frac{N}{N-f+2}, \nu^* \frac{N}{N-f+\nu^*-1}, \frac{N}{N-f}$, respectively, $\nu^* = \min(f+1,\nu)$. And ABD and erasure-coding refer to upper bounds achieved in \cite{ABD} and \cite{DGL, Cadambe_Lynch_Medard_Musial_NCA,androulaki2014_separate_metadata, CT},respectively, which corresponds to $f+1$ and $\nu \frac{N}{N-f}$.}
\label{fig_compare}
\end{figure}

\vspace{-5pt}
\subsection{Comparison with Prior Storage Cost Lower Bounds} 
\vspace{-5pt}
We compare our work with results of \cite{chockler2015space, chockler2007amnesic, LDR, spiegelman2015space, Wang_Cadambe, Wang_Cadambe_ISIT} which provide some impossibility results in connection to consistent shared memory emulation. The main reference that directly pertains to our work here is \cite{spiegelman2015space}, which describes interesting, non-trivial lower bounds on shared memory emulation algorithms where the server and client storage schemes satisfy certain restrictions. Reference \cite{spiegelman2015space} assumed that every bit stored in the system is associated uniquely with a write operation, and showed that under such a storage scheme, the worst case total storage cost of the system is at least $\Omega(\min(f,\nu) \log|\mathcal{V}|)$. The implication of \cite{spiegelman2015space} is that in the worst case, if the degree of concurrency is infinite and the server storage scheme is restricted in a particular manner, then the replication based algorithms of \cite{ABD, LDR} are approximately optimal.

	
The assumption of \cite{spiegelman2015space} that every bit stored is associated with a unique write operation is restrictive and does not apply to all possible storage methods. To see this, consider a scenario where $\mathcal{V}$ is a finite field. Let $v_1, v_2 \in \mathcal{V}$ be values corresponding to two different write operations. Suppose in some algorithm $A$, at some point of an execution, a server stores $v_1 + v_2$, where $+$ denotes the addition operator over the field. Then a bit stored by the  server cannot be uniquely associated with any of the write operations in an unambiguous way. Therefore, the proof technique and the result of \cite{spiegelman2015space} fails to provide any insight on the storage cost of the algorithm $A.$ Put differently, the storage method of \cite{spiegelman2015space} does not allow for server storage techniques that potentially compress the values of different versions together (See Appendix \ref{app:discussion} for more technical details). In contrast, the results of Theorems \ref{thm:second}, \ref{thm:third} are universal and would automatically apply to algorithm $A$. The result of Theorem \ref{thm:fourth} does not impose any structure on the storage method, and could also apply to algorithm $A$ if its write protocol satisfies the appropriate restrictions. 

References \cite{LDR, chockler2015space} describe impossibility results which are peripherally related to our work. In particular, the results of \cite{LDR, chockler2015space} show that if the readers or writers do not help write another client's value \cite{LDR}, or if the servers are modeled as read-write objects \cite{chockler2015space}, the number of servers must be at least linear in the degree of concurrency in the system. However, the results of \cite{LDR, chockler2015space} do not directly relate to the total-storage cost incurred by the algorithm. Reference \cite{chockler2007amnesic} considered algorithms where the readers do not change the state of the servers, that is, the readers do not write any values or metadata. The reference showed that for the class of algorithms considered, if the value comes from an \emph{infinite} set, then there exists no regular shared memory emulation algorithm that tolerates even a single server failure. The reference nonetheless does not provide any insight into the storage cost, particularly when the values come from a finite domain. 

In information theory literature, recently developed formulations generalize the classical erasure coding model with the goal of understanding storage costs in systems where consistency is important \cite{Wang_Cadambe_ISIT, Wang_Cadambe}. The models of \cite{Wang_Cadambe_ISIT, Wang_Cadambe} however, differs from the model considered here. In particular, the models of \cite{Wang_Cadambe, Wang_Cadambe_ISIT} does not involve formal notions of write and client protocols, and the decoding requirement is only loosely based on the notion of atomicity. Our proofs of Theorem \ref{thm:second}, \ref{thm:third} and \ref{thm:fourth} bear resemblance to storage cost lower bounds of \cite{Wang_Cadambe}.

In our concluding section, Section \ref{sec:conclusion}, we provide a summary of the state of the art based on the main results of our paper and of \cite{spiegelman2015space}.

\section{Preliminaries}
\label{sec:preliminaries}
We study the emulation of a shared atomic memory in an asynchronous message passing network. Our setting consists of a set of $N$ server nodes and a possibly infinite set of client nodes. Without loss of generality, we let the set of server nodes be $\{1,2,\ldots,N\}$. We denote the set of client nodes as $\mathcal{C}$. We assume a {multi-writer single-reader}\footnote{The storage cost lower bounds presented in our paper apply trivially to multi-writer single-reader shared memory emulation algorithms as well.} setting, that is, we assume that $\mathcal{C}$ has a single write client; the remaining clients in $\mathcal{C}$ are read clients. Each client node is connected to all the server nodes, and the servers are connected to each other via point-to-point reliable, asynchronous, channels. 
We assume that the readers receive read requests (invocations) from some external source and respond with object values. We assume that writers receive write-requests and respond with acknowledgements. Every new invocation at a client waits for a response of a preceding invocation at the same client. The goal of a shared memory emulation algorithm $A$ studied in this paper is to design the client and server protocols that implement a read-write register of a data object which can take values from a finite set $\mathcal{V}$, with the following safety and liveness properties. 

\emph{Safety Properties:}
{The algorithm must emulate a SWSR regular registers \cite{Lamport86}. Informally a regular shared memory object requires that every read operation returns either the value written by the latest write operation that terminates before the invocation of the read operation, or the value of a write operation that overlaps with the read operation. In Section \ref{sec:fourth} we consider multi-writer single-reader algorithms, and we require the algorithm to be atomic \cite{Lamport86}\footnote{In fact, we will study weakly regular multi-writer single-reader algorithms \cite{shao2011multiwriter}. See details in Section \ref{sec:fourth}}. Informally, in an atomic algorithm, the observed external behavior of every execution looks like the execution of a serial variable type. Recall that SWSR execution of an atomic shared memory emulation is also regular, so our lower bounds for regular algorithms in Theorems \ref{thm:second} and \ref{thm:third} also apply to atomic emulation algorithms.



\emph{Liveness Properties:}  An operation of a non-failed client must terminate in a fair execution, so long as some conditions are satisfied in the execution. Specifically, we require operations to terminate if the number of server failures in the execution is bounded by a parameter $f.$ {We consider algorithms with weaker liveness properties as well in Section \ref{sec:fourth}, where we ensure termination of operations in executions if the number of \emph{active} write operations is bounded. A formal statement of the weaker liveness properties is provided in Section \ref{sec:fourth}. }

We require the above correctness properties to hold irrespective of the number of client failures. The data object can take values from a {finite} set $\mathcal{V}.$

\subsection*{Storage Cost Definition}
Informally speaking, the {\em storage cost} of an algorithm is the total number of bits stored by the servers. In general, for an algorithm where the state of server node $i \in \{1,2,\ldots,N\}$ can take values from a set $\mathcal{S}_i,$ we define the storage cost of the server to be equal to $\log_{2} |\mathcal{S}_i|$ bits. The \emph{max-storage cost} of the algorithm $A$ is defined to be $$MaxStorage(A)=\displaystyle\max_{i \in \{1,2,\ldots,N\}} \log_{2}|\mathcal{S}_{n}|.$$ The \emph{total-storage cost} of the algorithm is defined to be $$TotalStorage(A)={\sum_{i =1}^{N} \log_{2}|\mathcal{S}_{i}|}.$$

\section{Storage Cost Lower Bound for Algorithms Without Gossip}
\label{sec:secondthm}
In Appendix \ref{sec:firstthm} we provide a simple but non-trivial proof of the storage cost lower bound that is analogous to Singleton bound. Some of the proof techniques there are also applied in this section. The readers can first read Appendix \ref{sec:firstthm} as an warm-up exercise.

Our main result of this section is a storage cost lower bound, assuming that servers do not gossip. Specifically, in this section, we assume that every message is sent from a client to a server, or from a server to a client.  The lower bound is an implication of Theorem \ref{thm:second}, which describes constraints on the cardinalities of the server states that must be satisfied by any atomic shared memory emulation algorithm where servers do no gossip. The lower bounds on the max- and total-storage costs are stated in Corollary \ref{cor:second}. After stating Theorem \ref{thm:second} and Corollary \ref{cor:second}, we provide an informal description of the proof of Theorem \ref{thm:second}, followed by a formal description.

\subsection{Statement of Theorem \ref{thm:second}}
\begin{theorem}
Let $A$ be a single-writer-single-reader shared memory emulation algorithm that implements a {regular} read-write object whose values come from a finite set $\mathcal{V}$.
Suppose that in $A$, every message is sent from a server to a client, or from a client to a server. Also, suppose that every server's state belongs to a set $\mathcal{S}$ in algorithm $A$. 

Suppose that the algorithm $A$ satisfies the following liveness property: 
In a fair execution of $A,$ if the number of server failures is no bigger than $f$, $f \ge 2$, then every operation invoked at a non-failing client terminates.

Then, for every subset $\mathcal{N} \subset \{1,2,\ldots,N\}$ where $|\mathcal{N}| = N-f$, 

$$\sum_{n \in \mathcal{N}} \log_{2}|\mathcal{S}_{i}| + \max_{n \in \mathcal{N}} \log_{2}|\mathcal{S}_{i}|  \geq {\log_{2}|\mathcal{V}| + \log_{2}\left(|\mathcal{V}|-1)\right) - \log_{2}(N-f)}.$$
\label{thm:second}
\end{theorem}

\begin{corollary}
Let $A$ be a single-writer-single-reader shared memory emulation algorithm that implements a {regular} read-write object whose values come from a finite set $\mathcal{V}$.
Suppose that in $A$, every message is sent from a server to a client, or from a client to a server. Also, suppose that every server's state belongs to a set $\mathcal{S}$ in algorithm $A$. 

Suppose that the algorithm $A$ satisfies the following liveness property: 
In a fair execution of $A,$ if the number of server failures is no bigger than $f$, $f \ge 2$, then every operation invoked at a non-failing client terminates. Then
$$MaxStorage(A) \geq \frac{\log_{2}|\mathcal{V}| + \log_{2}(|\mathcal{V}|-1) - \log_{2}(N-f)}{N-f+1},$$
and
$$TotalStorage(A) \geq \frac{N(\log_{2}|\mathcal{V}| + \log_{2}(|\mathcal{V}|-1)- \log_{2}(N-f))}{N-f+1}.$$
\label{cor:second}
\end{corollary}

\begin{proof}[Proof of Corollary \ref{cor:second}]
	We assume, without loss of generality, that $|\mathcal{S}_{1}| \leq |\mathcal{S}_{2} \leq \ldots \leq |\mathcal{S}_{N}|.$ From Theorem \ref{thm:second}, we have 
	$$\sum_{n =1}^{N-f} \log_{2} |\mathcal{S}_n| + \log_{2}|\mathcal{S}_{N-f}| \geq \log_{2}|\mathcal{V}| + \log_{2}(|\mathcal{V}|-1) - \log_{2}(N-f) .$$

	As a consequence, we have $\log_{2}|\mathcal{S}_{N-f}|\geq \frac{\log_{2} |\mathcal{V}| + \log_{2}(|\mathcal{V}|-1) - \log_2(N-f) }{N-f+1}$. Therefore, we have $\displaystyle\max_{n \in \{1,2,\ldots,N\}} \log_{2}|\mathcal{S}_{n}| \geq  \frac{\log_{2} |\mathcal{V}| + \log_{2}(|\mathcal{V}|-1) - \log_2(N-f)}{N-f+1}.$ Furthermore, we have $\log_{2}|\mathcal{S}_{n}| \geq \frac{\log_{2} |\mathcal{V}| + \log_{2}(|\mathcal{V}|-1) - \log_2(N-f) }{N-f+1}$ for every $n \in \{N-f+1,\ldots,N\}$. This implies the following chain of relations.
	\begin{eqnarray*}
		\sum_{n =1}^{N} \log_{2} |\mathcal{S}_n| &\geq& \log_{2}|\mathcal{V}| + \log_{2}(|\mathcal{V}|-1)- \log_{2}(N-f) - \log_{2}|\mathcal{S}_{N-f}| + \sum_{n=N-f+1}^{N} \log_{2}|\mathcal{S}_{n}| \\
		&\geq& \log_{2}|\mathcal{V}| + \log_{2}(|\mathcal{V}|-1)- \log_{2}(N-f) + \sum_{n=N-f+2}^{N} \log_{2}|\mathcal{S}_{n}| \\
		&\geq& \left(\log_{2}|\mathcal{V}| + \log_{2}(|\mathcal{V}|-1)- \log_{2}(N-f)\right) \left(1+ \frac{f-1}{N-f+1}\right)  
		\\&=& N \left(\frac{\log_{2} |\mathcal{V}| + \log_{2}(|\mathcal{V}|-1) - \log_2(N-f) }{N-f+1} \right)
	\end{eqnarray*}
This completes the proof.
\end{proof}

%

\subsection{Informal Proof Sketch of Theorem \ref{thm:second}} 
Informally speaking, our lower bound argument is as follows. For every subset $\mathcal{N} \subset \{1,2,\ldots,N\}$ where $|\mathcal{N}| = N-f$, we construct an execution where the servers in $\{1,2,\ldots,N\}-\mathcal{N}$ fail {at the beginning of the execution. The execution has two write operations for values $v_1$ and $v_2,$ where $v_1 \neq v_2$. The second write operation which writes value $v_2$ begins after the termination of the first write operation, which has value $v_1$.} 

{In this execution, after the point of termination of the first write, a reader can return $v_1$ because of {regularity}. Similarly, after the termination of the second write operation, a reader can return $v_2$. Therefore, the value $v_1$ is returnable from the servers at a point before the invocation of the second write operation and $v_2$ is returnable from the servers at a point after {the completion of} second write operation. Furthermore, at every point in the interval of the second write operation, at least one of $v_1$ or $v_2$ are returnable. This implies that, in the interval of the second write operation, there are two consecutive points $P$ and $P'$ such that $v_1$ must be returnable from the non-failing servers at point $P$ and $v_2$ must be returnable from the non-failing servers at point $P'$. Since $(v_1,v_2)$ can be any ordered pair of distinct values from $\mathcal{V}$,} there must be a one-to-one mapping between the set $\{(v_1, v_2): (v_1, v_2) \in \mathcal{V}\times \mathcal{V}, v_1 \neq v_2\}$ and the set of possible {configurations of server states} at points $P$ and $P'.$ This implies that the {number of possible server states} at points $P$ and $P'$ is at least $(|\mathcal{V}|)(|\mathcal{V}|-1)$. Since $P$ and $P'$ are {consecutive}, at most one non-failing server changes its state between these two points. At least $N-f-1$ non-failing servers have the same state at point $P$ as at point $P'.$ We use this fact to show that the number of elements in the set of possible server states at two {consecutive} points is at most $\prod_{n \in \mathcal{N}} |\cS_i| \times \max_{n \in \mathcal{N}}|\mathcal{S}_{n}| \times (N-f)$. Therefore, we get $\prod_{n \in \mathcal{N}} |\cS_i| \times \max_{n \in \mathcal{N}}|\mathcal{S}_{n}| \times (N-f) \geq (|\mathcal{V}|)(|\mathcal{V}|-1),$ which implies the lower bound. We now present a formal proof of the lower bound.

\subsection{Formal Proof of Theorem \ref{thm:second}}
 
Consider an arbitrary subset $\mathcal{N} \subset \{1,2,\ldots,N\}$ such that $|\mathcal{N}|=N-f$. {We construct} $|\mathcal{V}| \times (|\mathcal{V}|-1)$ executions of the algorithm $A$. In particular, for every tuple $(v_1, v_2) \in \mathcal{V}\times \mathcal{V}$ where $v_1 \neq v_2$, we create an execution $\alpha^{(v_1, v_2)}$ of algorithm $A$. In our proof, we first describe execution $\alpha^{(v_1,v_2)}$ in Section \ref{subsubsec:description}. Then, we describe some properties of execution $\alpha^{(v_1,v_2)}$ in Section \ref{subsubsec:properties}. 
We use the results of Section \ref{subsubsec:properties} to prove Theorem \ref{thm:second} in Section \ref{subsubsec:proof}.

\subsubsection{Execution $\alpha^{(v_1,v_2)}$}
\label{subsubsec:description}
In execution $\alpha^{(v_1,v_2)}$ the readers and the channels from and to the readers do not {perform} any actions. Among the set of write clients $\mathcal{C}_{w}$, only one write client takes actions. {The $f$ servers in $\{1,2,\ldots,N\} - \mathcal{N}$ fail at the beginning of the execution}. No further server failures occur in the execution.  A write $\pi_1$ is invoked at a write client with value $v_1$. All the components of the system except the readers, and the channels from and to the readers, take turns in a fair manner until the completion of $\pi_1$. Recall that, in a fair execution where {the} number of server failures is at most $f$, any write that begins eventually terminates irrespective of the number of read client failures. {From the perspective of the servers, write client and the channels between them, }the execution $\alpha^{(v_1,v_2)}$ is indistinguishable from a fair execution where all the read clients fail, we can ensure the execution can be extended until the write operation $\pi_1$ terminates. {Immediately after the termination of $\pi_1$}, a write operation $\pi_2$ with value $v_2$ begins. All the components of the system except the readers and the channels from and to the readers take turns in a fair manner until the completion of $\pi_2$. The execution $\alpha^{(v_1,v_2)}$ ends after the termination of $\pi_2$. 

\begin{figure}
	\begin{center}
		\includegraphics[height=1.4in, width=3in]{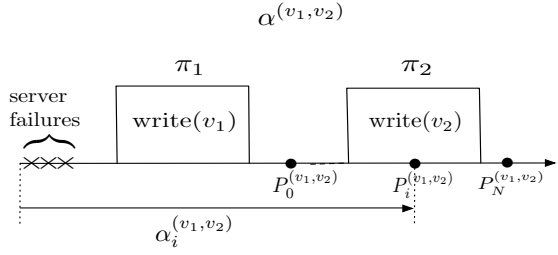}
\end{center}
\caption{Pictorial description of executions $\alpha^{(v_1,v_2)}$ and $\alpha^{(v_1,v_2)}_{i}$.}
\label{fig:executions}
\end{figure}

\subsubsection{Properties of Execution $\alpha^{(v_1,v_2)}$}
\label{subsubsec:properties}

Let $P_0^{(v_1, v_2)}, P_1^{(v_1, v_1)}, P_2^{(v_1,v_2)}, \ldots, P_M^{(v_1,v_2)}$ be the adjacent points (or points after successive steps) of the constructed execution $\alpha^{(v_1,v_2)}$, where $P_0^{(v_1,v_2)}$ is an arbitrary point after the termination of $\pi_1$ and before the invocation of $\pi_2$ and $P_M^{(v_1,v_2)}$ is an arbitrary point after the point of termination of $\pi_2$, and $M$ is a positive integer. For $i \in \{0,1,2,\ldots,M\}$, we denote by $\alpha_i^{(v_1,v_2)},$ the execution between the initial point of $\alpha^{(v_1,v_2)}$ and point $P_i^{(v_1,v_2)}$. The executions $\alpha^{(v_1,v_2)}$ and $\alpha_i^{(v_1,v_2)}$ are {depicted} in Fig. \ref{fig:executions}.

 For an integer $i$ in $\{0,1,\ldots,M\}$, we refer to point $P_{i}^{(v_1,v_2)}$  as a \emph{$k$-valent} point if it satisfies certain properties that are described in Definition \ref{def:1valent}, $k=1,2$. Informally speaking, a point $P_i^{(v_1,v_2)}$ is said to be $k$-valent if there exists an execution that starts at $P_i^{(v_1,v_2)},$ where a reader returns $v_k$. 

{
\begin{definition}[$k$-valent, $k \in \{1,2\}$]
\label{def:1valent}
For $i \in \{0,1,2,\ldots,M\}$, a point $P_i^{(v_1,v_2)}$ in the constructed $\alpha_i^{(v_1,v_2)}$  is said to be $k$-valent if we can extend $\alpha_i^{(v_1,v_2)}$ to an execution $\beta$ as follows: 
After $P_i^{(v_1,v_2)}$ all the messages from and to the writer are delayed indefinitely. A read operation starts at point ${P}_i^{(v_1,v_2)}$ and all the components, except the writer and the channels from and to the writer, {perform actions until the read operation terminates.} The read operation returns $v_k$. 
\end{definition}
}

{It should be noted that a point of an execution can be both 1-valent and 2-valent; thus our definition of valency has a somewhat different structure compared to other definition of valency in other impossibility arguments (e.g. \cite{FLP}).}

\begin{lemma}
For $i \in \{0,1,2,\ldots,M\}$, a point $P_i^{(v_1,v_2)}$ that is not $1$-valent is $2$-valent.
\label{lem:not1valent}
\end{lemma}

To show Lemma \ref{lem:not1valent}, we first prove Lemma \ref{lem:returnsv1orv2} which informally states that a reader that begins after the termination of the write $\pi_1$ should return $v_1$ or $v_2$ because of {regularity} of the algorithm. 

{\begin{lemma}
\label{lem:returnsv1orv2}
Consider an execution $\beta$ which is an extension of $\alpha_{i}^{(v_1, v_2)}$. In $\beta$, after point $P_{i}^{(v_1, v_2)}$, the writer stops taking steps and all messages from and to the writer are delayed indefinitely. A read operation begins at some point after point $P_{i}^{(v_1, v_2)}$ and terminates in $\beta$. 

Then, the read operation returns either $v_1$ or $v_2$.
\end{lemma}}
{The lemma is a natural consequence of the {regularity} of algorithm $A$. We provide a formal proof next.}
\begin{proof}
{The read operation is invoked after the termination of write operation $\pi_1$ in execution $\beta$. It is possible that the write operation $\pi_2$ is invoked before the invocation of the read operation. Because the algorithm is {regular}, we must be able to serialize operations in $\beta$. Because $\pi_2$ is invoked after the completion of $\pi_1$, the operation $\pi_2$ is serialized after operation $\pi_1$. Regarding the serialization point of the read operation, there are only two possibilities: (i) read operation is serialized immediately after $\pi_1$ and before $\pi_2$, and (ii) the read operation is serialized after $\pi_2$. If possibility (i) occurs, that is, if the read operation is serialized after $\pi_1,$ then it returns $v_1,$ which is the value of the write operation $\pi_1$. If possibility (ii) occurs, then the read returns $v_2$, which is the value of the write operation $\pi_2$. Therefore, the read operation returns either $v_1$ or $v_2$. This completes the proof.}
\end{proof}

\begin{proof}[{Proof of Lemma \ref{lem:not1valent}}]
Consider a point $P_i^{(v_1,v_2)}$ that is not $1$-valent. We show that it is $2$-valent by constructing an execution $\beta$ that satisfies Definition \ref{def:1valent} {for $k=2$}. The execution $\beta$ is an extension of $\alpha_i^{(v_1,v_2)}.$ 
In $\beta$, after point $P_i^{(v_1,v_2)}$ all the messages from and to the writer are delayed indefinitely. 

A read operation $\pi$ starts at point ${P}_i^{(v_1,v_2)}$ in $\beta$ and all the components, except the writer and the channels from and to the writer, execute their protocols taking turns in a fair manner. From the perspective of the servers, readers and the channels between the servers and the reader, the execution $\beta$ is indistinguishable from a fair execution of the algorithm where the write client fails before sending or receiving the messages in its channels. Because of the liveness properties satisfied by the algorithm, the read operation terminates. Because the read operation is invoked at point $P_i^{(v_1,v_2)}$ which is after the point of termination of the write operation $\pi_1$, Lemma \ref{lem:returnsv1orv2} implies that the read operation returns $v_1$ or $v_2$. However, because the point $P_{i}^{(v_1, v_2)}$ is not $1$-valent, the read operation cannot return $v_1$. Therefore the read returns $v_2$ at some point $Q$. Let $\beta$ denote the extension of $\alpha_i^{(v_1,v_2)}$ to the point $Q$. The execution $\beta$ satisfies the conditions of Definition \ref{def:1valent}{for $k=2$}. Therefore, the point $P_i^{(v_1,v_2)}$ is $2$-valent. 
\end{proof}

\begin{lemma}
There exists some integer $i \in \{0,2,\ldots,M-1\}$ such that $P_i^{(v_1,v_2)}$ is $1$-valent and $P_{i+1}^{(v_1,v_2)}$ is not $1$-valent.
\label{lem:adjacentpoints}
\end{lemma}
\begin{proof}
To show the lemma, we argue that the following two statements are true. 
\begin{enumerate}[(i)]
\item Point $P_0^{(v_1,v_2)}$ is $1$-valent.
\item Point $P_M^{(v_1,v_2)}$ is not $1$-valent. 
\end{enumerate}
Among all the numbers in $\{0,1,2,\ldots,M\},$ let $i$ denote the largest number such that $P_{i}^{(v_1,v_2)}$ is $1$-valent. If (i) is true, we note that the number $i$ exists. If (ii) is true, then $i < M$. Since $i$ is the largest number such that $P_{i}^{(v_1,v_2)}$ is $1$-valent, we infer that $P_i^{(v_1,v_2)}$ is $1$-valent, but $P_{i+1}^{(v_1,v_2)}$ is not $1$-valent. The point $P_{i}^{(v_1,v_2)}$ therefore satisfies the statement of the lemma. So, to show the statement of the lemma, it suffices to show (i) and (ii). We show (i) and (ii) formally next.

To show (i)  we extend $\alpha_0^{(v_1,v_2)}$ to an execution $\beta$ as per Definition \ref{def:1valent} for $1$-valency.
In $\beta$, after point $P_0^{(v_1,v_2)}$ all the messages from and to the writer are delayed indefinitely. 

Note that at point $P_0^{(v_1, v_2)}$ the write operation $\pi_2$ has not yet begun.
A read operation $\pi$ starts at point ${P}_0^{(v_1,v_2)}$ in execution $\beta$ and all the components, except the writer and the channels from and to the writer, execute their protocols taking turns in a fair manner. Note that the execution is indistinguishable from a fair execution of the algorithm where the write client fails before sending or receiving the messages in its channels. Because of the liveness properties satisfied by the algorithm, the read operation terminates. Note that there is only one write operation $\pi_1$ in execution $\beta$. Because the algorithm is {regular}, and because the read operation is invoked after the termination of $\pi_1,$ it is serialized after $\pi_1$. Therefore the read returns $v_1$ at some point $Q$. Let $\beta$ denote the extension of $\alpha_0^{(v_1,v_2)}$ to the point $Q$. The execution $\beta$ satisfies the conditions of Definition \ref{def:1valent} for $1$-valency.

To show (ii)  we show that we cannot extend $\alpha_M^{(v_1,v_2)}$ to an execution $\tilde{\beta}$ as per Definition \ref{def:1valent}. We provide a proof by contradiction. Suppose we can construct $\tilde{\beta}$ as per Definition \ref{def:1valent}. Note that in $\tilde{\beta},$ the read operation is invoked after point $P^{(v_1,v_2)}_M,$ which is after the point of termination of $\pi_2$. Therefore the read operation is invoked after the point of termination of the write $\pi_2$. Furthermore, $\pi_2$ is invoked after the point of termination of $\pi_1$. Because $\tilde{\beta}$ has {regular} operations, write operation $\pi_2$ is serialized after write operation $\pi_1$ and the read is serialized after the write $\pi_2.$ Therefore the read should return $v_2$, which is the value of write operation $\pi_2$. However, because $\tilde{\beta}$ satisfies Definition \ref{def:1valent}, the read returns $v_1$ which is not equal to $v_2$. Therefore the execution $\tilde{\beta}$ of algorithm $A$ does not have {regular} operations. This is a contradiction to the assumption that the algorithm $A$ is {regular}.  Therefore point $P_M^{(v_1,v_2)}$ cannot be $1$-valent.
This completes the proof.
\end{proof}

Next, we present the definition of a \emph{pair of critical points} of execution $\alpha^{(v_1,v_2)}$.
\begin{definition}[Critical points]
 Let $Q_1, Q_2$ be two points in execution $\alpha^{(v_1,v_2)}$. The pair of points $(Q_1, Q_2)$ is defined to be a pair of critical points if there exists a number $i$ in $\{0,2,\ldots,M-1\}$ such that 
\begin{itemize}
\item $Q_1 = P_{i}^{(v_1, v_2)},$ $Q_{2}=P_{i+1}^{(v_1,v_2)},$
\item $Q_1$ is $1$-valent,
\item $Q_2$ is not $1$-valent.
\end{itemize}
\label{def:criticalpoints}
\end{definition}

Lemma \ref{lem:adjacentpoints} implies that every execution $\alpha^{(v_1, v_2)}$ has at least one pair of critical points.
Lemma \ref{lem:not1valent} implies that if $(Q_1, Q_2)$ is a pair of critical points in $\alpha^{(v_1, v_2)},$ then point $Q_2$ is $2$-valent in $\alpha^{(v_1,v_2)}$. We need the following lemma before proceeding to prove Theorem \ref{thm:second}.

\begin{lemma}
Let $(Q_1,Q_2)$ be a pair of critical points of execution $\alpha^{(v_1,v_2)}$. Then,
\begin{enumerate}[(a)]
\item the readers, and the channels between the readers and the servers, are all in the same state at point $Q_2$ as at point $Q_1$;
\item there is at most one non-failing server $s$ such that its state at $Q_1$ is different from its state at $Q_2$.
\end{enumerate}
\label{lem:criticalpoints}
\end{lemma}
\begin{proof}
In execution $\alpha^{(v_1,v_2)}$, the readers and the channels between readers and servers do not {perform} any actions. So these components are in their initial state at every point of the execution, including points $Q_1$ and $Q_2$. This implies that statement (a) is true. We prove (b) next.
Note that $Q_{1}$ and $Q_{2}$ are adjacent points of the execution $\alpha^{(v_1,v_2)}$. There are three possibilities: (I) {a channel performed an action} between points $Q_1$ and $Q_{2}$, (II) {a server performed an action} between points $Q_{1}$ and $Q_{2}$, or (III) {a client performed an action} between points $Q_1$ and $Q_2$. We study these three possibilities separately.

{\bf{Case I}: A channel action took place between points $Q_1$ and $Q_2$.} Note that the algorithm $A$ does not send any messages on the channels between servers. The channels between readers and servers do not {perform} any actions in $\alpha^{(v_1,v_2)}$. Therefore, we only need to consider the case where a channel between a server and the writer takes an action. If a channel from the writer to server $s$ takes an action, then, for  every server in the set $\mathcal{N} - \{s\},$ there was no input, internal or output action between $Q_1$ and $Q_2$. Therefore every server in the set $\mathcal{N} - \{s\},$ has the same state at $Q_1$ as at $Q_2$. Similarly, if a channel from a server $s$ to the writer takes an action, then, for  every server in the set $\mathcal{N} - \{s\},$ there was no input, internal or output action between $Q_1$ and $Q_2$. Therefore every server in the set $\mathcal{N} - \{s\},$ has the same state at $Q_1$ as at $Q_2$. This completes the proof for Case I.

{\bf{Case II}: A server action took place between points $Q_1$ and $Q_2.$ } Let $s$ be the server that took an action between points $Q_1$ and $Q_2$. This implies that for every server in $\mathcal{N} - \{s\},$ no input, output or internal action was taken between these points. Therefore every server in $\mathcal{N} - \{s\},$ has the same state at $Q_1$ as at $Q_2$. This completes the proof for Case II.

{\bf{Case III:}  A client action took place between points $Q_1$ and $Q_2.$} If a client action takes place between $Q_1$ and $Q_2$, then, for every server in the system, no input, internal or external action was taken between points $Q_1$ and $Q_2.$ Therefore, in this case, every server has the same state at point $Q_1$ as at point $Q_2$. This completes the proof.
\end{proof}

We are now ready to prove Theorem \ref{thm:second}. 

\subsubsection{Proof of Theorem \ref{thm:second}}
\label{subsubsec:proof}
\emph{Proof of Theorem \ref{thm:second}.}
Lemma \ref{lem:adjacentpoints} implies that we can find a pair of critical points $(Q_1^{(v_1, v_2)}, Q_2^{(v_1, v_2)})$ in execution $\alpha^{(v_1, v_2)}$. From Lemma \ref{lem:criticalpoints}, we note that there is at most one non-failing server that changes state between $Q_1^{(v_1,v_2)}$ and $Q_2^{(v_1,v_2)}$. Let $s$ denote {the} server which changes state between points $Q_{1}^{(v_1,v_2)}$ and $Q_{2}^{(v_1,v_2)}$, {if there is one; if not}, let $s$ denote an arbitrary non-failing server. For any $s' \in \mathcal{N},$ if $s' \neq s,$ the state of the server $s'$ is the same at points $Q_{1}^{(v_1,v_2)}$ and $Q_{2}^{(v_1,v_2)}$.  

Let $\vec{S}^{(v_1,v_2)}$ be an element of $\prod_{n \in \mathcal{N}} \mathcal{S}_{n} \times \mathcal{N} \times \cup_{n \in \mathcal{N}} \mathcal{S}_{n}$ as follows. The first $N-f$ components of $\vec{S}^{(v_1,v_2)}$ denote the states of the $N-f$ {servers} in $\mathcal{N}$ at point $Q_{1}^{(v_1,v_2)}.$ The $(N-f+1)$st component of $\vec{S}^{(v_1,v_2)}$  denote the server index $s,$ and the $N-f+2$nd component is the state of server $s$ at point $Q_{2}^{(v_1,v_2)}$. Note that the number of elements in the set $\displaystyle\bigcup_{(v_1, v_2)\in \mathcal{V}\times \mathcal{V}, v_1 \neq v_2}\{\vec{S}^{(v_1, v_2)}\}$ is at most $\prod_{i\in \mathcal{N}} |\cS_i| \times (N-f) \times \max_{i\in \mathcal{N}} |\cS_i|$.

To prove Theorem \ref{thm:second}, we show that, if $(v_1, v_2)$ and $ (v_1',v_2')$ are two \emph{distinct} elements of the set $\{(x,y): (x,y) \in \mathcal{V}\times \mathcal{V}, x \neq y \},$ then $\vec{S}^{(v_1,v_2)} \neq \vec{S}^{(v_1',v_2')}.$ If we show this, {then} it implies that the number of elements in the set  $\displaystyle\bigcup_{(v_1, v_2)\in \mathcal{V}\times \mathcal{V},v_1 \neq v_2 }\{\vec{S}^{(v_1, v_2)}\}$ is at least equal to the number of elements in the set $\{(x,y): (x,y) \in \mathcal{V}\times \mathcal{V}, x \neq y\},$ which is equal to $(|\mathcal{V}|)\times (|\mathcal{V}|-1)$. This leads to the following chain of {inequalities}:
\begin{align*}
&\prod_{n\in \mathcal{N}} |\cS_n| \times (N-f) \times \max_{n\in \mathcal{N}} |\cS_n| \geq (|\mathcal{V}|)\times (|\mathcal{V}|-1) \\  
&\max_{n \in \mathcal{N}} \log_{2}|\mathcal{S}_{i}| + \sum_{n \in \mathcal{N}} \log_{2}|\mathcal{S}_{i}| \geq \log_{2}|\mathcal{V}| + \log_{2}\left(|\mathcal{V}|-1)\right) - \log_{2}(N-f)
\end{align*}
which implies the theorem. Therefore, to prove the theorem, it suffices to show that, if $(v_1, v_2) \neq (v_1',v_2'),$ then $\vec{S}^{(v_1,v_2)} \neq \vec{S}^{(v_1',v_2')}.$  Suppose, for contradiction, there are two {distinct} tuples $(v_1, v_2)$ and $(v_1', v_2')$ in $\{(x,y): (x,y) \in \mathcal{V}\times \mathcal{V}, x \neq y\}$ and $\vec{S}^{(v_1,v_2)} = \vec{S}^{(v_1',v_2')}.$

{Let $i$ denote an integer such that $Q_{1}^{(v_1, v_2)}$ is the point $P_{i}^{(v_1, v_2)}$ in $\alpha^{(v_1,v_2)}.$ We now construct executions $\beta_{1}^{(v_1, v_2)}$ and $\beta_2^{(v_1,v_2)},$ which are extensions of $\alpha_{i}^{(v_1,v_2)}$ and $\alpha_{i+1}^{(v_1, v_2)}$. Because the point $Q_{1}^{(v_1, v_2)}$ is $1$-valent, we know there an execution $\beta_{1}^{(v_1, v_2)}$ that extends $\alpha_i^{(v_1,v_2)}$ such that, after point $Q_{1}^{(v_1, v_2)}$ in $\beta_1^{(v_1, v_2)}$, all the messages from and to the writer are delayed indefinitely, and a read operation begins and returns $v_1$. Similarly, because the point $Q_{2}^{(v_1, v_2)}$ is $2$-valent, we know there an execution $\beta_{2}^{(v_1, v_2)}$ that extends $\alpha_{i+1}^{(v_1,v_2)}$ such that, after point $Q_{2}^{(v_1, v_2)},$ all the messages from and to the writer are delayed indefinitely, and a read operation begins and returns $v_2$.   

The following claim describes a useful property of executions $\beta_1^{(v_1,v_2)}$ and $\beta_2^{(v_1, v_2)}$.
\begin{claim}
Let $\vec{S}^{(v_1, v_2)}= \vec{S}^{(v_1', v_2')}.$  Consider the composite automaton $\mathcal{A}$ formed by the servers, the readers and the channels between the readers and servers. For $k \in \{1,2\}$, every component of the automaton $\mathcal{A}$ has the same state at point $Q_{k}^{(v_1,v_2)}$ in $\beta_k^{(v_1, v_2)}$ as at point $Q_{k}^{(v_1',v_2')}$ in execution $\beta_k^{(v_1',v_2')}.$ 
\label{lem:same_state}
\end{claim}

\emph{Proof of Claim \ref{lem:same_state}.}
	We first consider the case where $k=1$. At points $Q_{1}^{(v_1,v_2)}$ in $\beta_1^{(v_1, v_2)}$ and $Q_1^{(v_1',v_2')}$ in $\beta_1^{(v_1',v_2')}$, all the channels between the readers and servers are empty, the readers are in their initial state and the servers in $\{1,2,\ldots,N\}-\mathcal{N}$ have failed. Denoting $\mathcal{N}=\{a_1, a_2,\ldots,a_{N-f}\}$ where $a_1 < a_2 < \ldots a_{N-f}$, the state of every non-failing server $a_j, j \in \{1,2,\ldots,N-f\}$ at points $Q_{1}^{(v_1, v_2)}$ in $\beta_1^{(v_1,v_2)}$ and $Q_{1}^{(v_1', v_2')}$ in $\beta_1^{(v_1', v_2')}$ is respectively equal to the $j$th component of $\vec{S}^{(v_1,v_2)}$ and $\vec{S}^(v_1',v_2')$. Because $\vec{S}^{(v_1, v_2)} = \vec{S}^(v_1' ,v_2'),$  every non-failing server is at the same state at $Q_{1}^{(v_1, v_2)}$ in $\beta_1^{(v_1,v_2)}$ as at $Q_{1}^{(v_1', v_2')}$ in $\beta_1^{(v_1', v_2')}.$ This completes the proof for the case where $k=1$.

Consider the case where $k=2$. Let $s$ denote the server index determined by the $N-f+1$st component of $\vec{S}^{(v_1, v_2)}$, which is also equal to the server index determined by the $N-f+1$st component of $\vec{S}^{(v_1', v_2')}$. All the channels, readers and failed servers have the same state at points $Q_{2}^{(v_1,v_2)}$ in $\beta_2^{(v_1, v_2)}$ as at $Q_2^{(v_1',v_2')}$ in $\beta_2^{(v_1',v_2')}$.  The state of every non-failing server except server $s$ at points $Q_{2}^{(v_1, v_2)}$ in $\beta_2^{(v_1,v_2)}$ and $Q_{2}^{(v_1', v_2')}$ in $\beta_2^{(v_1', v_2')}$ is respectively determined by the corresponding component of $\vec{S}^{(v_1,v_2)}$ and $\vec{S}^(v_1',v_2')$. The state of server $s$ at points $Q_{2}^{(v_1, v_2)}$ in $\beta_2^{(v_1,v_2)}$ and $Q_{2}^{(v_1', v_2')}$ in $\beta_2^{(v_1', v_2')}$ are respectively determined by the $N-f+2$nd component of $\vec{S}^{(v_1,v_2)}$ and $\vec{S}^(v_1',v_2')$. Because $\vec{S}^{(v_1, v_2)}$ is equal to $\vec{S}^(v_1',v_2')$, every non-failing server is at the same state at $Q_{2}^{(v_1, v_2)}$ in $\beta_2^{(v_1,v_2)}$ as at $Q_{2}^{(v_1', v_2')}$ in $\beta_2^{(v_1', v_2')}.$ This completes the proof of the claim.

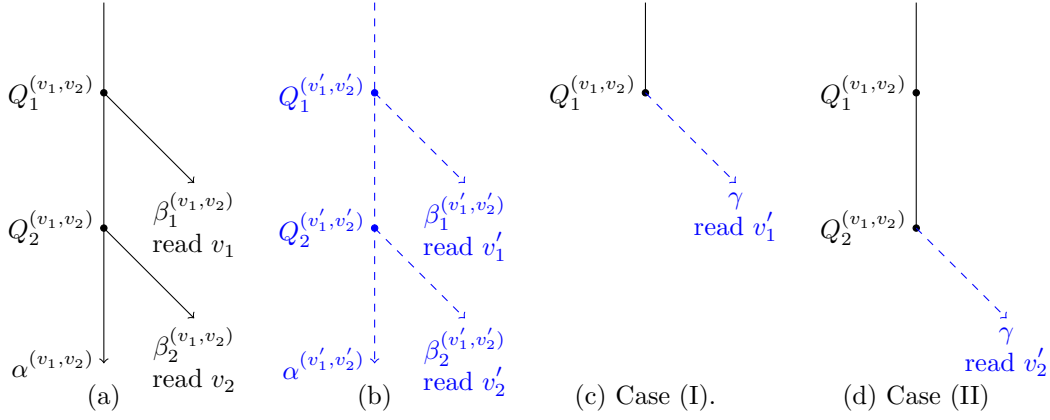
\begin{figure}
\centering
\begin{tikzpicture}[scale=0.6] 
\draw[->] (0,8) -- (0,0) node[left] {$\alpha^{(v_1,v_2)}$};
\draw[->] (0,6) -- (2,4) node[align=center,below] {$\beta_1^{(v_1,v_2)}$\\ read $v_1$};
\draw[->] (0,3) -- (2,1) node[align=center,below] {$\beta_2^{(v_1,v_2)}$\\read $v_2$};
\filldraw 
(0,6) circle (2pt) node[left] {$Q_1^{(v_1,v_2)}$} 
(0,3) circle (2pt) node[left] {$Q_2^{(v_1,v_2)}$};   
\draw (0,0) node[below,align=center]{\\(a)};

\draw[->, dashed, blue] (6,8) -- (6,0) node[left] {$\alpha^{(v_1',v_2')}$};
\draw[->, dashed, blue] (6,6) -- (8,4) node[align=center,below] {$\beta_1^{(v_1',v_2')}$\\ read $v_1'$};
\draw[->, dashed, blue] (6,3) -- (8,1) node[align=center,below] {$\beta_2^{(v_1',v_2')}$\\read $v_2'$};
\filldraw[blue] 
(6,6) circle (2pt) node[left] {$Q_1^{(v_1',v_2')}$} 
(6,3) circle (2pt) node[left] {$Q_2^{(v_1',v_2')}$};    
\draw (6,0) node[below,align=center]{\\(b)};

\draw (12,8) -- (12,6);
\filldraw 
(12,6) circle (2pt) node[left] {$Q_1^{(v_1,v_2)}$};
\draw[->,dashed,blue] (12,6) -- (14,4) node[align=center,below] {$\gamma$\\ read $v_1'$};
\draw (12,0) node[below,align=center]{\\(c) Case (I).};

\draw (18,8) -- (18,3);
\filldraw 
(18,6) circle (2pt) node[left] {$Q_1^{(v_1,v_2)}$}
(18,3) circle (2pt) node[left] {$Q_2^{(v_1,v_2)}$}; 
\draw[->,dashed,blue] (18,3) -- (20,1) node[align=center,below] {$\gamma$\\ read $v_2'$};
\draw (18,0) node[below,align=center]{\\(d) Case (II)};
\end{tikzpicture}
\caption{Depiction of the proof of Theorem \ref{thm:second}. (a) Executions $\beta_k^{(v_1,v_2)}, k=1,2$. (b) Executions $\beta_k^{(v_1',v_2')}, k=1,2$. (c) Constructed execution $\gamma$ for the case that $v_1' \notin \{v_1,v_2\}$. (d) Constructed execution $\gamma$ for the case that $v_2' \neq v_2$.  }
\label{fig:2ndbound}
\end{figure}

\emph{Proof of Theorem \ref{thm:second} continued.}
{We now use Claim \ref{lem:same_state} to obtain a contradiction.} Because $(v_1,v_2)$ and $(v_1', v_2')$ are distinct ordered pairs, there are only two possibilities:  (I) $v_1' \neq v_1, v_1' \neq v_2$, (II)  $v_2' \neq v_1, v_2' \neq v_2,$ or $v_2' = v_1, v_1' = v_2$, both of which imply that $v_2' \neq v_2$. We handle these possibilities separately (See Figure \ref{fig:2ndbound}). 

{\bf Case (I): $v_1' \neq v_1, v_1' \neq v_2$.} 

We create an execution $\gamma$ of the algorithm $A$ which contradicts Lemma \ref{lem:returnsv1orv2}. Let $i$ be an integer such that $Q_{1}^{(v_1,v_2)} = P_{i}^{(v_1,v_2)}$ in execution $\alpha^{(v_1,v_2)}$. The execution $\gamma$ extends execution $\alpha_i^{(v_1,v_2)},$ that is, it follows execution $\alpha^{(v_1,v_2)}$ until point $Q_{1}^{(v_1,v_2)}$. After point $Q_{1}^{(v_1, v_2)},$ the writer, and the channels from and to the writers do not perform any actions. After point $Q_1^{(v_1,v_2)},$ the servers, readers and the channels between the servers and readers in $\gamma$ follow the same steps as the corresponding components in $\beta_1^{(v_1',v_2')}.$ 

{Claim \ref{lem:same_state} implies that $\gamma$ is an execution of algorithm $A$.} {In particular, $\gamma$ is an extension of $\alpha_{i}^{(v_1, v_2)},$ where, after point $P_i^{(v_1, v_2)},$ the writer and channels from and to the writer do not perform any actions. From point $P_i^{(v_1,v_2)}$ onward, since the readers in $\gamma$ follow the steps of $\beta_1^{(v_1',v_2')}$ until completion, a read begins in $\gamma$ after this point and terminates returning $v_1',$ which is not equal to either $v_1$ or $v_2$. However, according to Lemma \ref{lem:returnsv1orv2}, the read operation in $\gamma$ should return either $v_1$ or $v_2,$. This is a contradiction. }Therefore, if $v_1 \neq v_1'$ and $v_1 \neq v_1'$, then $\vec{S}^{(v_1, v_2)} \neq \vec{S}^{(v_1',v_2')}$. 

{\bf Case (II): $v_2' \neq v_2$.} 

We create an execution $\gamma$ of the algorithm $A$ which leads to a contradiction. Let $i$ be an integer such that $Q_{1}^{(v_1,v_2)} = P_{i}^{(v_1,v_2)}$ in execution $\alpha^{(v_1,v_2)}$. The execution $\gamma$ extends execution $\alpha_{i+1}^{(v_1,v_2)},$ that is, it follows execution $\alpha^{(v_1,v_2)}$ until point $Q_{2}^{(v_1,v_2)}$. At point $Q_{2}^{(v_1, v_2)},$ the messages from the writers are delayed indefinitely. After point $Q_2^{(v_1,v_2)},$ the servers, readers and the channels between the servers and readers in $\gamma$ follow the same steps as the corresponding components in $\beta_2^{(v_1',v_2')}.$ 

{Claim \ref{lem:same_state} implies that $\gamma$ is an execution of algorithm $A$.} {In particular, $\gamma$ is an extension of $\alpha_{i+1}^{(v_1, v_2)},$ where, after point $P_{i+1}^{(v_1, v_2)},$ the writer and channels from and to the writer do not perform any actions. From point $P_{i+1}^{(v_1,v_2)}$ onward, since the readers in $\gamma$ follow the steps of $\beta_2^{(v_1',v_2')}$ until completion, a read begins in $\gamma$ after this point and terminates returning $v_2',$ which is not equal to $v_2$. However, according to Lemma \ref{lem:not1valent} and the fact that $Q_{2}^{(v_1,v_2)}$ is not $1$-valent, $Q_{2}^{(v_1, v_2)}$ is $2$-valent, and the read operation in $\gamma$ should return $v_2$. This is a contradiction. }Therefore, $\vec{S}^{(v_1, v_2)} \neq \vec{S}^{(v_1',v_2')}$.

This completes the proof.
\qed

\section{A Universal Storage Cost Lower Bound}
\label{sec:thirdthm}
Our main result of this section is a storage cost lower bound that is applicable to any {regular} shared memory emulation algorithm, even if it uses server gossip. The lower bound is an implication of Theorem \ref{thm:third}, which describes constraints on the cardinalities of the server states that must be satisfied by any atomic shared memory emulation algorithm. The lower bounds on the worst case and {total} storage costs are stated in Corollary \ref{cor:third}. We provide a sketch of the proof of Theorem \ref{thm:third}, highlighting the main differences from the proof of Theorem \ref{thm:second}.

\subsection{Statement of Theorem \ref{thm:third}}
\begin{theorem}
Let $A$ be a single-writer-single-reader shared memory emulation algorithm that implements a {regular} read-write object whose values come from a finite set $\mathcal{V}$.
Suppose that every server's state belongs to a set $\mathcal{S}$ in algorithm $A$. 

Suppose that the algorithm $A$ satisfies the following liveness property: 
In a fair execution of $A,$ if the number of server failures is no bigger than $f$, then every operation invoked at a non-failing client terminates.
Then, 
$$2 \max_{n \in \mathcal{N}} \log_{2}|\mathcal{S}_{i}| + \sum_{n \in \mathcal{N}} \log_{2}|\mathcal{S}_{i}| \geq \log_{2}|\mathcal{V}| + \log_{2}\left(|\mathcal{V}|-1\right) - 2 \log_{2}(N-f)$$
\label{thm:third}
\end{theorem}

\begin{corollary}
	Let $A$ be a single-writer-single-reader shared memory emulation algorithm that implements a {regular} read-write object whose values come from a finite set $\mathcal{V}$.
Suppose that every server's state belongs to a set $\mathcal{S}$ in algorithm $A$.

Suppose that the algorithm $A$ satisfies the following liveness property:
In a fair execution of $A,$ if the number of server failures is no bigger than $f$, then every operation invoked at a non-failing client terminates.
Then,
$$ MaxStorage(A) \geq \frac{\log_2 |\mathcal{V}|+ \log_2 |\mathcal{V}-1| - 2 \log_{2}(N-f)}{N-f+2},$$
$$ TotalStorage(A) \geq  \frac{N(\log_2 |\mathcal{V}|+ \log_2 |\mathcal{V}-1| - 2 \log_{2}(N-f))}{N-f+2}.$$ 
\label{cor:third}

\end{corollary}
The proof of Corollary \ref{cor:third} is similar to the proofs of Corollary \ref{cor:second} in Section \ref{sec:secondthm} and Corollary \ref{cor:first} in Appendix \ref{sec:firstthm}, and is omitted.


\subsection{Informal Proof Sketch of Theorem \ref{thm:third}}

{The proof of Theorem \ref{thm:third} shares many common elements with the proof of Theorem \ref{thm:second}. The main difference  is that now we need to carefully handle the actions performed by the channels between servers. An aspect that distinguishes the proof of Theorem \ref{thm:third} from the proof of Theorem \ref{thm:second} is that their definitions of $k$-valent points. For ease of readability, we inherit the lemmas and definitions from the proof of Theorem \ref{thm:third} into this section, so that the proofs can be compared easily. We begin with a proof sketch of Theorem \ref{thm:third}.

As in our proof of Theorem \ref{thm:second}, for every subset $\mathcal{N} \subset \{1,2,\ldots,N\}$ where $|\mathcal{N}| = N-f$, we construct an execution $\alpha^{(v_1,v_2)}$ where the servers in $\{1,2,\ldots,N\}-\mathcal{N}$ fail {at the beginning of the execution. The execution has two write operations for values $v_1$ and $v_2,$ where $v_1 \neq v_2$. The second write operation with value $v_2$ begins after the termination of the first write operation with value $v_1$.} 

In this execution, after the point of termination of the first write, if we let the channels between servers deliver all the gossip messages, and then begin a read operation after the delivery of these messages, a reader can return $v_1$ because of {regularity}. Similarly, after the termination of the second write operation,  if we let the channels between servers deliver all the gossip messages, a reader can return $v_2$. This implies that, in the interval of the second write operation, there are two consecutive points $P$ and $P'$ as follows: \begin{itemize}
\item If at point $P$ we stop the writers and the channels from the writers and let the channels between servers deliver all the gossip messages to arrive at point $Q$, then $v_1$ can be returned by a read operation that begins at point $Q$.
\item If at point $P'$ we stop the writers and the channels from the writers and let the channels between servers deliver all the gossip messages to arrive at point $Q'$, then $v_2$ can be returned by a read operation that begins at point $Q'$. 
\end{itemize}
By constructing the executions such that gossip messages are delivered in the same order, we can ensure that after the delivery of the messages, at most $2$ servers differ in their states between points $Q$ and $Q'$. 
We use this fact to show that the number of elements in the set of possible server states at points $Q$ and $Q'$ points is at most $\prod_{n \in \mathcal{N}} |\cS_i| \times \max_{n \in \mathcal{N}}|\mathcal{S}_{n}| \times \max_{n \in \mathcal{N}}|\mathcal{S}_{n}| \times (N-f)^2$. Therefore, we get $\prod_{n \in \mathcal{N}} |\cS_i| \times (\max_{n \in \mathcal{N}}|\mathcal{S}_{n}|)^2 \times (N-f)^2 \geq (|\mathcal{V}|)(|\mathcal{V}|-1),$ which implies the lower bound. 
}

\subsection{Formal Proof  of Theorem \ref{thm:third}}
\subsubsection{Execution $\alpha^{(v_1,v_2)}$ and Its Properties}
Let $\mathcal{N}$ be an arbitrary subset of $\{1,2,\ldots,N\}$ with $N-f$ elements. Like the proof of Theorem \ref{thm:second}, we construct $|\mathcal{V}| \times (|\mathcal{V}|-1)$ executions of the algorithm $A$. In particular, for every tuple $(v_1, v_2) \in \mathcal{V}\times \mathcal{V}$ where $v_1 \neq v_2$, we create an execution $\alpha^{(v_1, v_2)}$ of algorithm $A$. The execution $\alpha^{(v_1, v_2)}$ is constructed in a manner that is essentially the same as Section \ref{subsubsec:description}. The $f$ servers in $\{1,2,\ldots,N\}-\mathcal{N}$ fail at the beginning of $\alpha^{(v_1, v_2)}.$ The execution  $\alpha^{(v_1, v_2)}$ has two complete write operations $\pi_1$ and $\pi_2$ with values $v_1$ and $v_2$, with $\pi_2$ being invoked after the termination of $\pi_1$. 

Similar to the proof of Theorem \ref{thm:second}, we let $P_0^{(v_1, v_2)}, P_1^{(v_1, v_1)}, P_2^{(v_1,v_2)}, \ldots, P_M^{(v_1,v_2)}$ be a sequence of consecutive points in execution $\alpha^{(v_1,v_2)}$, where $P_0^{(v_1,v_2)}$ is an arbitrary point after the termination of $\pi_1$ and before the invocation of $\pi_2$, and $P_M^{(v_1,v_2)}$ is an arbitrary point after the point of termination of $\pi_2$.  We denote by $\alpha_i^{(v_1,v_2)},$ the execution between the initial point of $\alpha^{(v_1,v_2)}$ and point $P_i^{(v_1,v_2)}$. 

The definition of $1$-valent and $2$-valent points are similar to Definitions \ref{def:1valent}, with the exception that we allow the channels the servers to deliver their messages before the invocation of the read operation. We provide a formal definition of $1$-valent and $2$-valent points next. 

\begin{definition}[$k$-valent, $k \in \{1,2\}$]
\label{def:1valent_gossip}
For $i \in \{0,1,2,\ldots,M\}$, a point $P_i^{(v_1,v_2)}$ in the constructed $\alpha_i^{(v_1,v_2)}$ is said to be $k$-valent if we can extend $\alpha_i^{(v_1,v_2)}$ to an execution $\beta$ as follows: 
After $P_i^{(v_1,v_2)}$ all the messages from and to the writer are delayed indefinitely. At $P_{i}^{(v_1, v_2)}$ all the channels between the servers act, delivering all their messages. After the delivery of the messages in the channels between the servers, a read operation starts and all the components, except the writer and the channels from and to the writer, {execute their protocols until the read operation terminates.} The read operation returns $v_k$. 
\end{definition}

Results analogous to Lemmas \ref{lem:not1valent},  \ref{lem:returnsv1orv2} and \ref{lem:adjacentpoints} in the Section \ref{sec:secondthm} hold, with the modified definition of $k$-valent points. We simply restate these lemmas without proofs here for the sake of completeness. The proofs are essentially identical to the proofs in Section \ref{sec:secondthm}.

\begin{lemma}
For $i \in \{0,1,2,\ldots,M\}$, a point $P_i^{(v_1,v_2)}$ that is not $1$-valent is $2$-valent.
\label{lem:not1valent_gossip}
\end{lemma}

{\begin{lemma}
\label{lem:returnsv1orv2_gossip}
Consider an execution $\beta$ which is an extension of $\alpha_{i}^{(v_1, v_2)}$. In $\beta$, after point $P_{i}^{(v_1, v_2)}$, the writer stops taking steps and all messages from and to the writer are delayed indefinitely. A read operation begins at some point after point $P_{i}^{(v_1, v_2)}$ and terminates in $\beta$. 

Then, the read operation returns either $v_1$ or $v_2$.
\end{lemma}}

\begin{lemma}
There exists some integer $i \in \{0,2,\ldots,M-1\}$ such that $P_i^{(v_1,v_2)}$ is $1$-valent and $P_{i+1}^{(v_1,v_2)}$ is not $1$-valent.
\label{lem:adjacentpoints_gossip}
\end{lemma}

Lemma \ref{lem:criticalpoints} requires some minor modifications to include the possibility that, between a pair of critical points, a channel between servers may perform an action. We state and prove the analogous lemma next.

We inherit the definition of critical points from Definition \ref{def:criticalpoints}. The only change is that the terms $1$-valent and $2$-valent points use Definition \ref{def:1valent_gossip}. We restate the definition here for the sake of completeness.
\begin{definition}[Critical points]
 Let $Q_1, Q_2$ be two points in execution $\alpha^{(v_1,v_2)}$. The pair of points $(Q_1, Q_2)$ is defined to be a pair of critical points if there exists a number $i$ in $\{0,2,\ldots,M-1\}$ such that 
\begin{itemize}
\item $Q_1 = P_{i}^{(v_1, v_2)},$ $Q_{2}=P_{i+1}^{(v_1,v_2)}$,
\item $Q_1$ is $1$-valent,
\item $Q_2$ is not $1$-valent.
\end{itemize}
\label{def:criticalpoints_gossip}
\end{definition}

\begin{lemma}
Let $(Q_1,Q_2)$ be a pair of critical points of execution $\alpha^{(v_1,v_2)}$. Then,
\begin{enumerate}[(a)]
\item the readers, and the channels between the readers and the servers, are all in the same state at point $Q_2$ as at point $Q_1$;
\item there is at most one non-failing server such that its state at $Q_1$ is different from its state at $Q_2$.

\item among all the channels between the servers, there is at most one channel whose state at $Q_1$ is different from its state at $Q_2$. 
\end{enumerate}
\label{lem:criticalpoints_gossip}
\end{lemma}
\begin{proof}
In execution $\alpha^{(v_1,v_2)}$, the readers and the channels between readers and servers do not perform any actions. So these components are in their initial state at every point of the execution, including points $Q_1$ and $Q_2$. This implies that statement $(a)$ is true. We prove $(b)$ and $(c)$ next.
Note that $Q_{1}$ and $Q_{2}$ are adjacent points of the execution $\alpha^{(v_1,v_2)}$. There are three possibilities: (I) {a channel performed an action} between points $Q_1$ and $Q_{2}$, (II) {a server performed an action} between points $Q_{1}$ and $Q_{2}$, or (III) {a client performed an action} between points $Q_1$ and $Q_2$. 

{\bf{Case I}:} The channels between readers and servers do not perform any actions in $\alpha^{(v_1,v_2)}$. We consider two sub-cases: (i) a channel between a server and the writer takes an action, or (ii) a channel between a server and another server takes an action. 

\emph{Case I (i):} If a channel from the writer to server $s$ takes an action, then, for  every server in the set $\mathcal{N} - \{s\},$ there was no input, internal or output action between $Q_1$ and $Q_2$. Therefore every server in the set $\mathcal{N} - \{s\},$ has the same state at $Q_1$ as at $Q_2$. If a channel from a server $s$ to the writer takes an action, then, for every server, there was no input, internal or output action between $Q_1$ and $Q_2$. Therefore every server in the set $\mathcal{N} - \{s\},$ has the same state at $Q_1$ as at $Q_2$. Furthermore, all channels between the servers are in the same state at $Q_1$ as at $Q_2$. This completes the proof in Case I (i).

\emph{Case I (ii):} If a channel from a server, say $s'$ to a server, say $s$ takes an action, between $Q_1$ and $Q_2,$ then, every server in the set $\mathcal{N} - \{s\},$ has the same state at $Q_1$ as at $Q_2$. Furthermore, all the channels between the servers except the channel from $s'$ to $s$ have the same state at $Q_2$ as at $Q_1$. Therefore $(b)$ and $(c)$ are satisfied in Case I.

\textbf{Case II:} The proof of statement $(b)$ is the same as the proof of Lemma \ref{lem:criticalpoints_gossip}. We show statement $(c)$ here. Suppose server $s$ takes an action between $Q_1$ and $Q_2$. If the action is not an output action, or if the server outputs a message onto a channel to the writer, then the states of all the channels between servers are the same at $Q_1$ and $Q_2$. If the action outputs a message on to the channel to another server, say $s'$, then all the channels between the servers except the channel from $s$ to $s'$ have the same state at $Q_2$ as at $Q_1$. Therefore statement (c) is true for Case II.

\textbf{Case III:} The proof of statement $(b)$ is the same as the proof of Lemma \ref{lem:criticalpoints_gossip}. Statement (c) holds because all the channels between the servers have the same state at $Q_1$ as at $Q_2$. 
\end{proof}

We are now ready to prove Theorem \ref{thm:third}. 

\subsubsection{Proof of Theorem \ref{thm:third}}
\emph{{Proof of Theorem \ref{thm:third}}.}
Lemma \ref{lem:adjacentpoints_gossip} implies that we can find a pair of critical points $(Q_1^{(v_1, v_2)}, Q_2^{(v_1, v_2)})$ in execution $\alpha^{(v_1, v_2)}$. From Lemma \ref{lem:criticalpoints_gossip}, we note that there is at most one non-failing server and one channel that changes its state between $Q_1^{(v_1,v_2)}$ and $Q_2^{(v_1,v_2)}$. Let $s$ denote {the} non-failing server which changes state between points $Q_{1}^{(v_1,v_2)}$ and $Q_{2}^{(v_1,v_2)}$, {if there is one; if not}, let $s$ denote an arbitrary non-failing server. Let $s'$ be a non-failing server such that the channel from $s''$ to $s'$ change its state between $Q_1^{(v_1, v_2)}$ and $Q_2^{(v_1, v_2)},$ for some server $s'',$  if there is such a server \footnote{Between $Q_1^{(v_1, v_2)}, Q_2^{(v_1, v_2)}$, if there is a server $s$ that changes its state, {and} a channel between two servers $s''$ and $s'$ that changes its state, it is easy to show that $s \in \{s', s''\}.$ We nonetheless use distinct notation for servers $s, s', s''$ since it simplifies presentation.}; let $s'$ be an arbitrary non-failing server if there is not. 

We know that $Q_{1}^{(v_1, v_2)}$ is the point $P_{i}^{(v_1, v_2)}$ for some $i \in \{0,1,\ldots,M-1\}$. We create executions $\beta_1^{(v_1, v_2)}$ and $\beta_2^{(v_1, v_2)},$ which are respectively extensions of $\alpha_{i}^{(v_1, v_2)}$ and $\alpha_{i+1}^{(v_1, v_2)}$ next. The construction of executions $\beta_{1}^{(v_1, v_2)}$ and $\beta_2^{(v_1,v_2)}$ has subtle, but important differences from corresponding constructions in the proof of Theorem \ref{thm:second} because, here, we carefully handle the actions of the channels between the servers. After presenting the constructions of $\beta_{1}^{(v_1, v_2)}$ and $\beta_2^{(v_1,v_2)},$ we state Claim $\ref{lem:same_state_gossip},$ which is analogous to Claim \ref{lem:same_state}.

Because $Q_{1}^{(v_1, v_2)}$ is a $1$-valent point, there exists an execution $\beta_{1}^{(v_1, v_2)}$ which is an extension of $\alpha_{i}^{(v_1, v_2)}$ such that, at point $P_{i}^{(v_1, v_2)},$ the writer and channels from the writer stop performing actions, the channels between the servers deliver all their messages, a read operation begins after the delivery of these messages and returns $v_{1}$. We denote by ${R}_{1}^{(v_1, v_2)}$ a point in $\beta_1^{(v_1, v_2)}$ after $P_{i}^{(v_1, v_2)},$ after the channels between the servers deliver all their messages, but before the read operation is invoked. 

We now create execution $\beta_{2}^{(v_1, v_2)}.$ The execution $\beta_{2}^{(v_1, v_2)}$ follows $\alpha^{(v_1, v_2)}$ until point $Q_{2}^{(v_1, v_2)}.$ At point $Q_2^{(v_1, v_2)},$ all the channels between the servers act delivering all their messages. For a server $j$ in $\{1,2,\ldots,N-f\}-\{s,s'\},$ the channels with destination $j$ are at the same state at $Q_2^{(v_1, v_2)}$ as they are at $Q_1^{(v_1, v_2)}$; these channels act and deliver their messages in the same order as they do after point $Q_1^{(v_1, v_2)}$ in $\beta_{1}(v_1, v_2).$ At point $Q_{2}^{(v_1, v_2)}$, server $j \in \mathcal{N}-\{s, s''\}$ is at the same state as it was at point $Q_1^{(v_1, v_2)}$. Also, at point $Q_{2}^{(v_1, v_2)},$ server $j$ receives messages in the same order as it does at point $Q_1^{(v_1, v_2)}$ in $\beta_1^{(v_1, v_2)}$; on receiving each message, server $j$ takes the same action in $\beta_2^{(v_1, v_2)}$ as it does in $\beta_1^{(v_1, v_2)}$. 
The channels with destinations $s$ or $s'$ deliver messages in some arbitrary order, and servers $s$ and $s'$ perform actions based on the protocol specified by algorithm $A$. We denote this point as $R_2^{(v_1, v_2)}$. It is worth noting that at point $R_2^{(v_1, v_2)},$ all the channels are empty, and every server in $\mathcal{N}-\{s,s'\}$ is at the same    state as it is at point $R_1^{(v_1, v_2)}$ in $\beta_1^{(v_1, v_2)}$. After point $R_2^{(v_1, v_2)}$, the writer and the channels from and to the writer do not perform any actions. At $R_{2}^{(v_1, v_2)}$, a read operation begins, all the components except the writer and the channels from and to the writer act in a fair manner until the read returns. Because the point $Q_{2}^{(v_1, v_2)}$ is $2$-valent but not $1$-valent, the read returns $v_2$ in $\beta_2^{(v_1, v_2)}$.

We now derive a lower bound on the storage cost by showing some properties on server states at points ${R}_{1}^{(v_1, v_2)}$ and ${R}_{2}^{(v_1, v_2)}$ in executions $\beta_{1}^{(v_1, v_2)}$ and $\beta_{2}^{(v_1, v_2)}$ respectively.

Let $\vec{S}^{(v_1,v_2)}$ be an element of $ \prod_{n \in \mathcal{N}} \mathcal{S}_{n} \times \mathcal{N} \times \displaystyle\cup_{n \in\mathcal{N}}\mathcal{S}_{n} \times \mathcal{N} \times \displaystyle\cup_{n \in\mathcal{N}}\mathcal{S}_{n}$ as follows. The first $N-f$ components of $\vec{S}^{(v_1,v_2)}$ denote the states of the $N-f$ {servers} in $\mathcal{N}$ at point $R_{1}^{(v_1,v_2)}.$ The $(N-f+1)$st component of $\vec{S}^{(v_1,v_2)}$ denotes the server index $s$ and $(N-f+2)$nd component denotes the state of server $s$ at point $R_{2}^{(v_1,v_2)}$ in $\alpha^{(v_1, v_2)}$. 
The $(N-f+3)$nd component denotes the server index $s'$ and the $(N-f+4)$th component denotes the state of server $s'$ at point $R_2^{(v_1, v_2)}$ in execution $\beta_2^{(v_1, v_2)}$. Note that the number of elements in the set $\displaystyle\bigcup_{(v_1, v_2)\in \mathcal{V}\times \mathcal{V}, v_1 \neq v_2 }\{\vec{S}^{(v_1, v_2)}\}$ is at most $\prod_{i\in \mathcal{N}} |\cS_i| \times (N-f)^2 \times (\max_{i\in \mathcal{N}} |\cS_i|)^2$.

To prove Theorem \ref{thm:third}, we show that, if $(v_1, v_2)$ and $ (v_1',v_2')$ are two \emph{distinct} elements of the set $\{(x,y): (x,y) \in \mathcal{V}\times \mathcal{V}, x \neq y \},$ then $\vec{S}^{(v_1,v_2)} \neq \vec{S}^{(v_1',v_2')}.$ If we show this, {then} it implies that the number of elements in the set  $\displaystyle\bigcup_{(v_1, v_2)\in \mathcal{V}\times \mathcal{V}, v_1 \neq v_2 }\{\vec{S}^{(v_1, v_2)}\}$ is at least equal to the number of elements in the set $\{(x,y): (x,y) \in \mathcal{V}\times \mathcal{V}, x \neq y\},$ which is equal to $(|\mathcal{V}|)\times (|\mathcal{V}|-1)$. This leads to the following chain of {inequalities}:
\begin{align*}
&\prod_{n\in \mathcal{N}} |\cS_n| \times (N-f)^2 \times (\max_{n\in \mathcal{N}} |\cS_n|)^2 \geq (|\mathcal{V}|)\times (|\mathcal{V}|-1) \\  
&2 \max_{n \in \mathcal{N}} \log_{2}|\mathcal{S}_{i}| + \sum_{n \in \mathcal{N}} \log_{2}|\mathcal{S}_{i}| \geq \log_{2}|\mathcal{V}| + \log_{2}\left(|\mathcal{V}|-1\right) - 2 \log_{2}(N-f)
\end{align*}
which implies the theorem. Therefore, to prove the theorem, it suffices to show that, if $(v_1, v_2) \neq (v_1',v_2'),$ then $\vec{S}^{(v_1,v_2)} \neq \vec{S}^{(v_1',v_2')}.$ The remainder of our proof is similar to Theorem \ref{thm:second}. We highlight the main differences.

Suppose, for contradiction, there are two {distinct} tuples $(v_1, v_2)$ and $(v_1', v_2')$ in $\{(x,y): (x,y) \in \mathcal{V}\times \mathcal{V}, x \neq y\}$ and $\vec{S}^{(v_1,v_2)} = \vec{S}^{(v_1',v_2')}.$ We now state a lemma that is analogous to Claim \ref{lem:same_state}.

\begin{claim}
Let $\vec{S}^{(v_1, v_2)}= \vec{S}^{(v_1', v_2')}.$  Consider the composite automaton formed by the servers, the readers, the channels between the servers, and the channels between the readers and servers. For $k \in \{1,2\}$, every component of this system at point $R_{k}^{(v_1,v_2)}$ in $\beta_k^{(v_1, v_2)}$ is identical to the state of the corresponding component at point $R_{k}^{(v_1',v_2')}$ in execution $\beta_k^{(v_1',v_2')}.$ 
\label{lem:same_state_gossip}
\end{claim}

\emph{Proof of Claim \ref{lem:same_state_gossip}.}
	We first consider the case where $k=1$. At points $R_{1}^{(v_1,v_2)}$ in $\beta_1^{(v_1, v_2)}$ and $R_1^{(v_1',v_2')}$ in $\beta_1^{(v_1',v_2')}$, all the channels between the servers, and the channels between the readers and servers are empty, the readers are in their initial state and the servers in $\{1,2,\ldots,N\}-\mathcal{N}$ have failed. Denoting $\mathcal{N}=\{a_1, a_2,\ldots,a_{N-f}\}$, The state of every non-failing server $s$ at points $R_{1}^{(v_1, v_2)}$ in $\beta_1^{(v_1,v_2)}$ and $R_{1}^{(v_1', v_2')}$ in $\beta_1^{(v_1', v_2')}$ is respectively equal to the $s$th component of $\vec{S}^{(v_1,v_2)}$ and $\vec{S}^(v_1',v_2')$. Because $\vec{S}^{(v_1, v_2)} = \vec{S}^(v_1' ,v_2'),$  every non-failing server is at the same state at $R_{1}^{(v_1, v_2)}$ in $\beta_1^{(v_1,v_2)}$ as at $Q_{1}^{(v_1', v_2')}$ in $\beta_1^{(v_1', v_2')}.$ This completes the proof for the case where $k=1$.

Now consider the case where $k=2$. Let $s$ and $s'$ respectively denote server indices determined by the $N-f+1$st and $N-f+3$rd components of $\vec{S}^{(v_1, v_2)}$. All the channels, readers and failed servers have the same state at points $R_{2}^{(v_1,v_2)}$ in $\beta_2^{(v_1, v_2)}$ as at $R_2^{(v_1',v_2')}$ in $\beta_2^{(v_1',v_2')}$.  
Denoting $\mathcal{N} = \{a_1,a_2,\ldots,a_{N-f}\}$, where $a_1 < a_2, \ldots < a_{N-f}$, the state of a non-failing server $a_j \in \{1,2,\ldots,N-f\}-\{s, s'\}$ at points $Q_{2}^{(v_1, v_2)}$ in $\beta_2^{(v_1,v_2)}$ and $Q_{2}^{(v_1', v_2')}$ in $\beta_2^{(v_1', v_2')}$ is respectively equal to the $j$th component of $\vec{S}^{(v_1,v_2)}$ and $\vec{S}^(v_1',v_2')$. The states of servers $s$ and $s'$ at point $Q_{2}^{(v_1, v_2)}$ in $\beta_2^{(v_1,v_2)}$ are respectively determined by the $N-f+2$nd and $N-f+4$th components of $\vec{S}^{(v_1,v_2)};$  the states of $s$ and $s'$ at $Q_{2}^{(v_1', v_2')}$ in $\beta_2^{(v_1', v_2')}$ are respectively determined by the $N-f+2$nd and $N-f+4$th components of $\vec{S}^{(v_1',v_2')}.$ Because $\vec{S}^{(v_1, v_2)}$ is equal to $\vec{S}^(v_1',v_2')$, every non-failing server has the same state  at $R_{2}^{(v_1, v_2)}$ in $\beta_2^{(v_1,v_2)}$ as at $R_{2}^{(v_1', v_2')}$ in $\beta_2^{(v_1', v_2')}.$ This completes the proof of the claim.

\emph{Proof of Theorem \ref{thm:third} continued.}
We now use Claim \ref{lem:same_state_gossip} to obtain a contradiction.
Because $(v_1, v_2)$ and $(v_1', v_2')$ are distinct, there are only two possibilities:  (I) $v_1' \neq v_1, v_1' \neq v_2$, (II)  $v_2' \neq v_1, v_2' \neq v_2$, or $v_2' = v_1, v_1' = v_2$, both of which imply that $v_2' \neq v_2$. We handle these possibilities separately.

{\bf Case (I): $v_1' \neq v_1, v_1' \neq v_2$.}
We create an execution $\gamma$ as follows. The execution $\gamma$ follows $\beta_1^{(v_1, v_2)}$ until point $R_1^{(v_1, v_2)}.$ From point $R_1^{(v_1, v_2)}$ in $\gamma$, every component except the writer and the channels from and to the writer follows the same steps of the corresponding component in $\beta_1^{(v_1', v_2')}$ from point $R_1^{(v_1', v_2')}$. Claim \ref{lem:same_state_gossip} implies that $\gamma$ is an execution of algorithm $A.$ In particular $\gamma$ extends $\alpha_i^{(v_1, v_2)}$ such that, the writer stops performing actions, messages from and to the writer are delayed indefinitely and a reader returns $v_1',$ which is not equal to $v_2$ or $v_1$. Therefore $\gamma$ violates Lemma \ref{lem:returnsv1orv2_gossip} and results in a contradiction. Therefore, if $v_1' \neq v_1, v_1' \neq v_2$, then $\vec{S}^{(v_1, v_2)} \neq \vec{S}^{(v_1', v_2')}$.

{\bf Case (II): $v_2' \neq  v_2$.} 
{We show that the point $Q_{2}^{(v_1, v_2)}$ is not $2$-valent by constructing an execution $\gamma$ that satisfies Definition 6.8 for $k=2$.} The execution $\gamma$ follows $\beta_2^{(v_1, v_2)}$ until point $R_2^{(v_1, v_2)}.$ From point $R_2^{(v_1, v_2)}$ in $\gamma$, every component except the writer and the channels from and to the writer follows the same steps of the corresponding component in $\beta_2^{(v_1', v_2')}$ from point $R_2^{(v_1', v_2')}$. Claim \ref{lem:same_state_gossip} implies that $\gamma$ is an execution of algorithm $A.$ In particular, $\gamma$ extends $\alpha_{i+1}^{(v_1, v_2)}$ such that, after point $Q_{2}^{(v_1, v_2)}$, the messages from and to the writer are delayed indefinitely, the channels between the servers deliver all their messages, and then, a reader begins returning $v_2',$ which is not $v_2$. However, by Lemma \ref{lem:not1valent_gossip} and that $Q_{2}^{(v_1, v_2)}$ is not $1$-valent, $Q_{2}^{(v_1, v_2)}$ is $2$-valent, and the read of $\gamma$ should return $v_2$, which is a contradiction.

This completes the proof.
\qed

\vspace{-5pt}
\section{Storage Cost Lower Bound for a Restricted Class of Algorithms}
\vspace{-5pt}
\label{sec:fourth}
In this section, we study a restricted class of algorithms where the write protocols have specific structure. In our restricted class, the write protocols consist of a fixed number of phases. In the protocols that we study, there is only one phase where a message containing information about the actual value is sent to the servers. The formal statement of our assumptions on the write protocol in Section \ref{sec:assumptions} is somewhat technically involved. However the write protocols of most previous algorithms \cite{DGL, AJX, CT, FAB, dobre_powerstore, Cadambe_Lynch_Medard_Musial_new,  Cadambe_Lynch_Medard_Musial_NCA} satisfy our assumptions. After stating our assumptions, we state in Theorem \ref{thm:fourth} in Section \ref{sec:fourththm}, a storage cost lower bound that applies to the class of algorithms that we study. The lower bound of Theorem \ref{thm:fourth} is much larger than the bound of Theorems \ref{thm:second} and \ref{thm:third}, and is close to the costs of previously developed algorithms. 

\vspace{-5pt}
\subsection{Protocol Assumptions}
\label{sec:assumptions}
\vspace{-5pt}
We now state three assumptions on the write protocol, Assumptions 1, 2 and 3.


\textbf{Assumption 1:} 
\emph{The state of a write client during a write operation is of the form $(v, m, h(v,m))$ where 
		 $v \in \mathcal{V}$ is {the value of the write operation}, 
		 $m$ is an element of a set $\mathcal{M}$, and
		 $h(m,v)$ is the value of function $h$ whose domain is $\mathcal{V}\times \mathcal{M}$ and range is a finite set.
	} 
	
	The set $\mathcal{M}$ is referred to as the {metadata set of the write protocol} of the algorithm. The function $h(m,v)$ can contain components of the send buffers that depend on the value, and hashed values used for verification to handle Byzantine adversaries \cite{androulaki2014_separate_metadata, HGR, dobre_powerstore}. To describe Assumption $2$, we first define the notion of a quorum system and a phase. A quorum system $\mathcal{Q}$ is a collection of subsets of $\{1,2,\ldots,N\}$. 
\begin{definition}[Phase]
For an arbitrary subset $\mathcal{N} \subseteq \{1,2,\ldots,N\}$  and a quorum system $\mathcal{Q},$ a $(\mathcal{N},\mathcal{Q})$-\emph{phase} consists of a sequence of actions at a write client as follows:
(i) Send message $m_{n}$ to server node $n$ for \emph{every} $n \in \mathcal{N}$.
(ii) Wait for responses from least one subset of servers in the collection $\mathcal{Q}.$
(iii) Perform internal actions, and finish the phase.
\vspace{-5pt}
\end{definition}
\begin{definition}[Decomposable into phases]
A write protocol is said to be \emph{decomposable into phases} if, 
on the invocation of a write operation, it invokes a phase, and  on the termination of a phase, it either invokes another phase, or terminates the write operation.
\vspace{-5pt}
\end{definition}

We are now ready to state Assumption $2$.  

\textbf{Assumption 2:} \emph{The write protocol is decomposable into phases.}

Before we state Assumption 3, we state the notion of the black-box action. Informally, a write client action is said to be a black-box action if every internal and output action of the client handles the data values as a black box, that is, actions treat the data object obliviously without regard to the actual value of the object. We state our assumption more formally next.
Recall that the write protocol is specified as a set of transitions $(\textrm{old-state, action, new-state}).$ 
\begin{definition}[Black-box Action]\label{def:black_box_action}
An internal or output action $\sigma$ performed by a write client is said to be a black-box action if the following holds: if, for some value $v \in \mathcal{V}$, 
\begin{itemize}
\item the action $\sigma$ is enabled when the client's state is $(m,v, h(m,v))$ for some $m \in \mathcal{M},$ and
\item the action $\sigma$ can result in the transition of the client's state from $(m,v,h(m,v))$ to $(m',v,h(m,v))$ for some $m' \in \mathcal{M},$ 
\end{itemize}
then, for every value $v' \in \mathcal{V}$, 
\begin{itemize}
 \item the action $\sigma$ is enabled when the client's state is $(m,v', h(m,v')),$ and
\item the action $\sigma$ can result in the transition of the client's state from $(m,v',h(m,v'))$ to $(m', v', h(m',v')).$ 
\end{itemize}
\end{definition}
For example, in the ABD algorithm \cite{ABD}, all actions are black-box actions. In particular, if the action of sending of a value is enabled when the metadata is $m$ for particular value $v$, then the send action is enabled for every value $v' \in \mathcal{V}$. Similarly, in erasure coding based algorithms, if the action of sending a codeword symbol is enabled in at one state, it is enabled at every state. 

Note that a write client's output actions are send and return. Send actions of a write client are categorized as value-dependent and value-independent actions.

\begin{definition}[Value-dependent and value-independent send actions]
A black-box send action $\sigma$ that is enabled during a write operation is said to be value-independent if the message sent does not depend on the value of the operation.  
A send action that is not a value-independent send action is referred to as a value-dependent send action.\end{definition}
For example, in the ABD algorithm, value-independent send actions involve sending query messages to the servers. Messages sent by value-dependent and value-independent send actions are respectively referred to as value-dependent and value-independent messages. We are now ready to state Assumption 3.

\textbf{Assumption 3:} \textbf{(a)} \emph{All write client actions are black-box actions,} and \textbf{(b)} \emph{in a write operation $\pi$ in an execution $\alpha$, if there is a phase where at least one value-dependent send action is performed, then every send action in every subsequent phase of the write operation $\pi$ is a value-independent send action.}

	In particular, Assumption 3(b) implies that there is at most one phase where the writer sends value-dependent messages on behalf of a write operation in any execution. We next state our main result.

\subsection{Statement of Theorem \ref{thm:fourth}}
\label{sec:fourththm}

We state our theorem for \emph{weakly regular} MWMR registers \cite{shao2011multiwriter}. Informally  a weakly regular shared memory object is one that supports concurrent write and read operations where, in every execution, for every terminating read operation $\pi_r$,  there is a subset $\Phi$ of the non-terminating write operations such that the operations in $\{\pi_r\}\cup\Phi \cup \Pi$ look like the execution of a serial variable, where $\Pi$ is the set of all terminating write operations in the execution. 

An atomic register is also a weakly regular register. Therefore, the storage cost of Theorem \ref{thm:fourth} applies for atomic registers as well. 

In an execution $\alpha$, a write operation $\pi$ is said to be active at point $P$ if the point $P$ is after the point of invocation but before the point of termination of $\pi$.

\begin{theorem}\label{thm:fourth}
Let $A$ be a multi-writer-single-reader shared memory emulation algorithm that implements a weakly regular read-write object whose values come from a finite set $\mathcal{V}$. Suppose algorithm $A$ satisfies Assumptions $1, 2$ and $3$ stated in Section \ref{sec:assumptions}, and following liveness property: 
In a fair execution of $A,$ if the number of server failures is no bigger than $f$ and the number of active write operations is no bigger than $\nu$, then every operation invoked at a non-failing client terminates.

Then, for every subset $\mathcal{N} \subseteq \{1,2,\ldots,N\}, |\mathcal{N}|=\min(N-f+\nu-1,N)$
{$$ \sum_{n\in\mathcal{N}} \log_{2} |\cS_n|\ge \log_{2} \binom{|\cV|-1}{\nu^*} - \nu^{*} \log_{2} (N-f+\nu^*-1) -\log_2(\nu^*!) $$}
where {$\nu^{*} = \min(\nu, f+1)$.}
\end{theorem}

\begin{corollary}\label{cor:fourth}
Let $A$ be a multi-writer-single-reader shared memory emulation algorithm that implements a weakly regular read-write object whose values come from a finite set $\mathcal{V}$. Suppose algorithm $A$ satisfies Assumptions $1, 2$ and $3$ stated in Section \ref{sec:assumptions}, and following liveness property: 
In a fair execution of $A,$ if the number of server failures is no bigger than $f$ and the number of write operations is no bigger than $\nu$, then every operation invoked at a non-failing client terminates. Then \begin{align*}
MaxStorage(A) &\ge  \frac{\nu^* }{N-f+\nu^*-1} \log |\cV| - o(\log|\cV|),\\
TotalStorage(A) &\ge  \frac{\nu^* N }{N-f+\nu^*-1} \log |\cV| - o(\log|\cV|),
\end{align*}
where $\nu^* = \min(\nu,f+1)$
\end{corollary}

The proof of Corollary \ref{cor:fourth} is similar to the proofs of Corollary \ref{cor:second} in Section \ref{sec:secondthm} and Corollary \ref{cor:first} in Appendix \ref{sec:firstthm}, and is omitted.


\subsection{Informal, Intuitive Proof Sketch for Theorem \ref{thm:fourth}}
\begin{figure}
\centering
\includegraphics[width=0.95\textwidth]{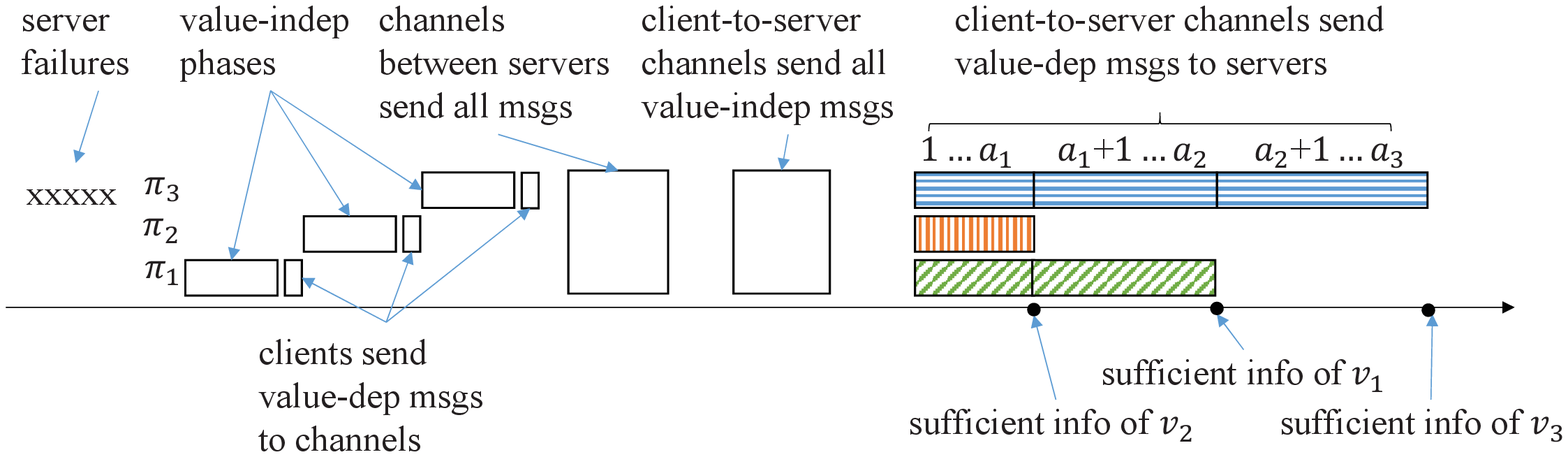}
\caption{Pictorial description of the execution for $\nu=3$. $b_1=2, b_2=1,b_3=3$.}
\label{fig:4thbound}
\end{figure}
For simplicity of exposition, assume $\mathcal{N}=\{1,2,\dots,N-f+2\}$. For our informal description, we set the parameter $\nu=3$, and $f \ge \nu-1=2$. Our proof of Theorem \ref{thm:fourth} constructs an execution $\alpha$ where the servers in $\{N-f+3, N-f+4, \ldots,N\}$ fail at the beginning of the execution. The execution has $\nu=3$ write operations $\pi_1, \pi_2, \pi_3$ with distinct values $v_1, v_2, v_3$ respectively invoked at distinct clients $ C_1, C_2, C_3$. We assume that at the beginning of the execution before the invocation of any write operation, a default initial value $v_0$ can be returned by any read operation, and that values $v_1, v_2, v_3$ are distinct from the default initial value $v_0$.

Writes $\pi_1, \pi_2,\pi_{3}$ are invoked respectively at clients $C_1,C_2, C_3.$  Recall that, as per Assumption $3$, there is at most one phase where the clients send value-dependent messages. Operations $\pi_1, \pi_2,\pi_3$ execute their protocols in a fair manner until they reach their respective phases where they send the value-dependent messages. The clients send the value-dependent messages onto the channels, but the channels do not yet deliver these value-dependent messages. Consider the point $P$ after all three clients send there value-dependent messages. At point $P$, the channels from the clients to the servers carry all the value-dependent messages that can be sent in the execution $\alpha$, and none of them are delivered to any of the servers.

Now we construct an execution $\alpha'$ which extends the execution $\alpha$ beyond point $P$ to a point $P'$ by allowing the channels from the clients to the servers act to deliver all the value-dependent messages to the first $N-f$ servers at point $P$. After the delivery of the messages, it must be the case that the first $N-f$ servers store ``sufficient information'' to return at least one of the values $v_1, v_2$ or $v_3$. This is because there are no additional phases where value-dependent messages are sent in $\alpha'$, and so, even if we extend the execution beyond point $P'$, the servers cannot receive any additional information related to $v_1, v_2$ or $v_3$. Furthermore, if we extend the execution $\alpha'$ beyond $P'$ by letting at least one of the operations $\pi_1, \pi_2, \pi_3$ complete by performing the remaining phases, then, because of weak regularity, at least one of the values  $v_1, v_2$ or $v_3$ must be returnable from the first $N-f$ servers after the completion of the operation. So it must be the case that at point $P'$ in $\alpha'$, the servers store sufficient information to return one of $v_1, v_2$ or $v_3$. Let $a_1$ be the smallest number such that, if the channels between the clients and the first $a_1$ servers act after point $P$ by delivering all their messages, then the first $a_1$ servers store sufficient information of value $v_{b_1}$ for some $b_1 \in\{1,2,3\}$. Note that $1\leq a_1 \leq N-f$. In our execution $\alpha$, at point $P$, we let all the channels deliver all their value-dependent messages to the first $a_1$ servers.  Denote the point after the delivery of the messages as $P_1$. Since we chose $a_1$ to be the smallest number of servers that contain sufficient information of any one of $v_1, v_2, v_3$, sufficient information of any \emph{one} of $v_1, v_2$ or $v_3$ is not contained from any of the first $a_1-1$ servers at point $P_1$.

Now we construct an execution $\alpha''$ as an extension of execution $\alpha$ beyond $P_1$ by allowing the channels from clients $\{C_1,C_2,C_3\}-\{C_{b_1}\}$ deliver their value-dependent messages to servers in $\{a_1+1, a_1+2,\ldots,N-f+1\}$. After the delivery of the messages, sufficient information of one of the values $v_{b_2}$ must be available in the first $N-f+1$ servers for some $b_2 \in \{1,2,3\} - \{b_1\}.$ This is because, if, after the delivery of the value-dependent messages in $\alpha''$, server $a_1$ stops taking actions, and clients $\{C_1,C_2,C_3\}-\{C_{b_1}\},$ and the first $N-f+1$ servers apart from server $a_1$ take actions in a fair manner, then one of the operations $\pi_1, \pi_2,\pi_3$ completes; weak regularity implies that one of the values $v_1, v_2, v_3$ is returnable from the first $N-f+1$ servers. However, note that sufficient information related to $v_{b_1}$ is not contained in the first $a_1-1$ servers. As a consequence, $v_{b_1}$ cannot be returnable from the first $N-f+1$ servers in $\alpha''$ if server $a_1$ does not take actions and we do not allow the value-dependent messages from client $C_{b_1}$ to be delivered to any one of the servers $\{a_1+1, a_1+2,\ldots, N-f+1\}$; therefore a value $v_{b_2} \neq v_{b_1}$ must be returnable from the first $N-f+1$ servers. Let $a_2$ be a number with $a_1 < a_2 \leq N-f+1$ such that, if all the channels deliver their value-dependent messages to the first $a_1$ servers and the channels from clients in $\{C_1, C_2, C_3\}-\{C_{b_1}\}$ deliver their value-dependent messages to the servers in $\{a_1+1, a_1+2,\ldots, a_2\},$ then sufficient information about value $v_{b_2}$ is contained in the first $a_2$ servers for some $b_2 \neq b_1, b_2 \in \{1,2,3\}$. In $\alpha$, at point $P_1$, we let  clients in $\{C_1, C_2, C_3\}-\{C_{b_1}\}$ deliver their value-dependent messages to the servers in $\{a_1+1, a_1+2,\ldots, a_2\}.$ Denote the point after the delivery of the messages as $P_2$.

 Similarly, if  we let the channels from remaining client in $\{C_1,C_2,C_3\}-\{C_{b_1},C_{b_2}\}$ deliver their value-dependent messages at point $P_2$ in $\alpha$ to the servers in $\{a_2+1, a_2+2,\ldots,N-f+2\}$, then sufficient information about the value in $\{v_1, v_2, v_3\}-\{v_{b_1},v_{b_2}\}$ is contained from the first $N-f+2$ servers after the delivery of the messages. At this point, sufficient information about all $3$ values is contained in the first $N-f+2$ servers. We can show that this implies that there is a one-to-one mapping from the states of the first $N-f+2$ servers to the values in $(\mathcal{V}-\{v_0\})^{3}$, where $v_0$ is the initial value. This implies that the storage cost must be at least $\frac{3}{N-f+2}\log_2|\mathcal{V}|+o(\log_{2}|\mathcal{V}|).$

{Our proof involves developing an appropriate notion of \emph{sufficient information of a value} that is applicable even when each server stores some arbitrary function of the values of the different versions it receives. In particular, we cannot directly borrow from other work \cite{spiegelman2015space}, whose notion of sufficient information of a value is tied to the storage scheme imposed by the model studied. Our notion of sufficient information is crystallized in a notion of valency that is more general as compared with Section \ref{sec:secondthm}. 

\subsection{Proof}
\label{sec:fourthproof}
To avoid cumbersome notation, we assume that $\nu \leq f+1$ and we prove Theorem \ref{thm:fourth} for $\mathcal{N}=\{1,2,\ldots,N-f+\nu-1\}$. The proof for the general case readily follows from the proof provided here.

To prove the lower bound, we construct a set of executions as follows. Every element of the set is parametrized by:
\begin{itemize}
	\item an arbitrary permutation $\sigma:\{1,2,\ldots,\nu\}\rightarrow \{1,2,\ldots,\nu\}$, 
	\item an arbitrary collection of numbers $a_1, a_2, \ldots, a_{\nu} \in \{0, 1,\ldots,N-f+\nu-1\},$ where $a_1 \leq a_2 \leq \ldots \leq a_{\nu},$ and 
	\item an arbitrary collection of distinct values $v_0, v_1, v_2,\ldots v_{\nu} \in \mathcal{V}$.
\end{itemize}

An execution parametrized by permutation $\sigma$, numbers $a_1, \ldots, a_{\nu}$ and values $v_0, v_1 \ldots, v_{\nu}$ is denoted by $\alpha^{(v_0, v_1,\ldots, v_{\nu})}$ $(\sigma, a_1, \ldots, a_{\nu}).$ We assume that $v_0$ indicates the default initial value that should be returned by a read operation in an execution where there is no write operation.

The proof is split into four parts. In the first part, we describe our construction of execution $\alpha^{(v_0, v_1,\ldots, v_{\nu})}(\sigma, a_1, \ldots, a_{\nu})$ in Section \ref{sec:thm4_construction}. We also prove a property, Lemma \ref{lem:metadata}, regarding the states of the components in execution $\alpha^{(v_0, v_1,\ldots, v_{\nu})}(\sigma, a_1, \ldots, a_{\nu}).$ In the second part, we define the notion of valency tailored to our class of executions in Section \ref{sec:thm4_valence}. In the third part, we prove a key property in Lemma \ref{lem2} in Section \ref{sec:thm4_lemma}. In the fourth and final part, we use Lemma \ref{lem2} to prove Theorem \ref{thm:fourth}. 

\textbf{Notation:} In the sequel, we denote $\vec{v} = {(v_0, v_1,\ldots, v_{\nu})},$ where  $v_0, v_1, v_2,\ldots v_{\nu} \in \mathcal{V}.$ We use the term \emph{value-vector} to refer to the $\vec{v}$. We assume that the components of a value vector are distinct, and $v_0$ is the default initial value.

\subsubsection{Description of Execution $\alpha^{\vec{v}}(\sigma, a_1, \ldots, a_{\nu})$}
\label{sec:thm4_construction}

We next describe the execution $\alpha^{\vec{v}}(\sigma, a_1, \ldots, a_{\nu})$ { for an arbitrary value vector $\vec{v}$ .} In the execution, only $\nu$ distinct write clients ${C}_{1},\ldots,{C}_{\nu} \in \mathcal{C}_{w}$ act. In particular, for $i \in \{1,2,\ldots, \nu\},$ one write operation $\pi_i$ is invoked at client $C_{i}$ with value $v_{i}$. For every collection of integers $a_1, a_2, \ldots, a_{\nu} \in \{0,1,\ldots,N-f+\nu-1\},$ where $a_1 \leq a_2 \leq \ldots \leq a_{\nu},$ and for every permutation $\sigma,$ the execution  $\alpha^{\vec{v}}(\sigma, a_1, \ldots, a_{\nu})$ is an extension of an execution $\alpha^{\vec{v}}_0$. The final point of $\alpha^{\vec{v}}_{0}$ is denoted as $P^{\vec{v}}_{0}$. We first describe $\alpha^{\vec{v}}_0$ and later describe $\alpha^{\vec{v}}(\sigma, a_1, \ldots, a_{\nu}).$ 

We choose arbitrarily a reference value vector $\vec{v}_{\textrm{ref}} \in \{v_0\}\times \mathcal{V}^{\nu}.$ The components of $\vec{v}_{\textrm{ref}}$ are denoted as $v_0, v_{1,\textrm{ref}}, v_{2,\textrm{ref}}, \ldots, v_{\nu,\textrm{ref}}.$ Next, we construct execution $\alpha^{\vec{v}_{\textrm{ref}}}_0.$ After that, we use our construction of $\alpha^{\vec{v}_{\textrm{ref}}}_0$ to describe the execution $\alpha^{\vec{v}}_0$ for an arbitrary value vector $\vec{v}$.

\begin{algorithm}
Execution $\alpha^{\vec{v}_{\textrm{ref}}}_0$\\
\noindent\rule{\textwidth}{0.1pt}
\begin{algorithmic}[1]
\State Initial point: all components are at their initial states. \label{l1}
\State The last $f+1-\nu$ servers fail \label{l2}
\For{$i=1$ to $\nu$} \label{l3}
\State Operation $\pi_i$ is invoked at client $C_i$ with value $v_{i,\textrm{ref}}.$ 
	\State Client $C_i,$ the non-failed servers and the channels take steps in a fair manner until the beginning of a phase $R_i$, where at least one value-dependent send action is enabled, or until operation $\pi_i$ terminates without sending any value-dependent message.
	\State If operation $\pi_i$ is not terminated, client $C_i$ performs the send actions corresponding to phase $R_i,$ sending all value-dependent message to the channels. (The channels from the client to the servers do not yet deliver the value-dependent messages. The client $C_i$ does not perform any more actions.)
	\EndFor

\State The channels between the servers deliver all their messages. 
\State The channels from the clients to the servers deliver all the value-independent messages. 
\label{l12}
\end{algorithmic}
\end{algorithm}

Note that in execution $\alpha^{\vec{v}_\textrm{ref}}_{0},$ until the beginning of phase $R_i$ at client $C_i$ for any $i \in \{1,2,\ldots,\nu\}$, every action is a value-independent action. The only value-dependent send actions in the execution are the send actions performed by clients $C_i$ in their corresponding phases $R_i$, for $i \in \{1,2,\ldots,\nu\}.$ Note that these value-dependent messages are not delivered in the execution $\alpha^{\vec{v}}_{0}$. Therefore, from the perspective of the servers, all the received messages by the final point $P_0^{\vec{v}_{\textrm{ref}}}$ are value-independent messages. 

We now describe execution $\alpha^{\vec{v}}_{0}$ for an arbitrary value vector $\vec{v}$. In execution $\alpha^{\vec{v}}_{0},$ the behavior of the environment, servers, clients, and is the same as in execution $\alpha.$ That is, for every point $P^{\textrm{ref}}$  in $\alpha^{\vec{v}_{\textrm{ref}}}_0,$ there is a corresponding point $P$ in $\alpha^{\vec{v}}_0.$ If at $P^{\textrm{ref}}$ there is a write invoked at client $C_i$ with value $v_{i,\textrm{ref}}$ for some $i \in\{1,2,\ldots\nu\}$, then at point $P$ in $\alpha^{\vec{v}},$ there is write invoked at client $C_i$ with value $v_i$. If at $P^{\textrm{ref}}$ a channel or a server performs an action, then at point $P$, the same channel or server performs the same action.  If at $P^{\textrm{ref}}$ a client performs a black-box action $\sigma$, then at point $P$, the same client performs the action $\sigma$ such that it takes the corresponding internal transition as in Definition \ref{def:black_box_action}. We next argue that $\alpha^{\vec{v}}_{0}$ is a valid execution of the algorithm. Since read clients do not act in execution $\alpha^{\vec{v}},$ we only argue that servers and write clients conform to their protocol specifications in the execution. 

\begin{lemma}
The execution $\alpha_0^{\vec{v}}$ is a valid execution of the algorithm for every value vector $\vec{v}.$ 
\end{lemma}
\begin{proof}

The servers, channels and clients in the system are I/O automata. Let $\mathcal{A}_{s}$ denote the automaton formed from the composition of all the servers and the channels between servers. Let $\mathcal{A}_{c}$ denote the automaton formed by the composition of all the write clients. 

The input to the automaton $\mathcal{A}_{s}$ are the messages delivered in the channels from the clients to the servers. Since only value-independent messages are delivered to the servers in executions $\alpha^{\vec{v}_{\textrm{ref}}}_0$ and $\alpha^{\vec{v}}_0,$ the inputs to $\mathcal{A}_{s}$ in the two executions are the same. Since the components of $\mathcal{A}_{s}$ follow their protocol specifications in execution $\alpha^{\vec{v}_{\textrm{ref}}}_0,$ and the steps of the components in  $\alpha^{\vec{v}}_0,$ are the same as in execution $\alpha^{\vec{v}_{\textrm{ref}}}_0$ the components of $\mathcal{A}_{s}$ follow their protocol specification in execution $\alpha^{\vec{v}}_{0}$ as well. 

The inputs to components of $\mathcal{A}_{c}$ are the messages sent from the servers to the write clients, and the write invocations. Because the write invocations and the server actions in execution $\alpha^{\vec{v}}_0,$ are the same as in execution $\alpha^{\vec{v}_{\textrm{ref}}}_0,$ the inputs to the automaton $\mathcal{A}_{c}$ are the same. Since all write client actions are black-box actions, and since write client steps follow their protocol specifications in $\alpha^{\vec{v}_\textrm{ref}}_{0},$ the clients follow their protocol specifications in execution $\alpha^{\vec{v}}_{0}$ as well. Therefore, execution $\alpha^{\vec{v}}_{0}$ is a valid execution of the algorithm.
\end{proof}

It is useful to note that, at the final point $P_0^{\vec{v}}$ of $\alpha^{\vec{v}}_0,$ the states of the servers, the channels amongst the servers, the channels from the servers to the clients, and the metadata components of the client states are all independent of the value vector. The only components whose state may depend on the value vector are the write clients, and the channels from the write clients to the servers. The following lemma describes this property more formally.

\begin{lemma}
\label{lem:schedule0}
For any two value vectors $\vec{v}, \vec{v}',$ the state of every server, every channel from a server to a server or client, and the metadata components of the state of any writer $C\in\{C_1, C_2, \ldots, C_{\nu}\}$ at the final point $P_0^{\vec{v}}$ of execution $\alpha_0^{\vec{v}}$ is the same as at the final point $P_0^{\vec{v}'}$ of $\alpha_0^{\vec{v}'}.$ 
\end{lemma}
\begin{proof}
Consider a component automaton $c$ the system, which is a server, or a channel from a servers to a server or a client. We consider four possibilities separately.

\emph{$c$ is channel from a server to another server}: At point $P_0^{\vec{v}}$ in $\alpha_0^{\vec{v}},$ as at point $P_0^{\vec{v}'}$ in $\alpha_0^{\vec{v}}$ the channel $c$ is empty. Therefore, for a component that is a channel between two servers, the lemma statement holds. 

\emph{$c$ is a server:} Server $c$ takes the same steps in $\alpha_0^{\vec{v}},$ as in $\alpha_0^{\vec{v}'}.$ Therefore the state of $c$ at point $P_0^{\vec{v}}$ in $\alpha_0^{\vec{v}}$ is the same as its state $P_0^{\vec{v}'}$ in $\alpha_0^{\vec{v}'}.$  

\emph{$c$ is channel from a server to a client}: If $c$ is a channel from a server $s$ to a client, note that the server $s$ takes the same steps in $\alpha_0^{\vec{v}},$ as in $\alpha_0^{\vec{v}'}.$ Therefore, the inputs to channel $c$ in the two executions are the same. The steps of the channel are the same in the two executions as well. Therefore, for a component that is a channel from a server to a client, the lemma statement holds. 

\emph{$c$ is a client $C_i$ for some $i \in \{1,2,\ldots,\nu\}$:} The client $C_i$ takes the same steps in $\alpha_0^{\vec{v}}$ as in $\alpha_0^{\vec{v}'},$ except for its final action in the execution. The final action of the execution is either a value-dependent send action, or an operation termination action. In either case, because all actions of client $C_i$ are black-box actions, and the client $C_i$ performs the same action that is performed in $\alpha_0^{\vec{v}_{\textrm{ref}}},$ the metadata component of the client $C_i$ after the final action is the same at point  $P_0^{\vec{v}}$ in $\alpha_0^{\vec{v}},$  as at point $P_0^{\vec{v}'}$ in $\alpha_0^{\vec{v}'}.$ 
\end{proof}

We describe $\alpha^{\vec{v}}(\sigma,a_1,\dots,a_{\nu})$ as an extension of $\alpha_0^{\vec{v}}$. If $a_1 \geq 1$, then at point $P_0^{\vec{v}}$ of execution $\alpha^{\vec{v}}(\sigma,a_1,\dots,a_{\nu})$, for every server node $n$ in $\{1,2,\ldots,a_1\}$, the channels from all the writers to server $n$ deliver their messages. We denote the point after the delivery of all messages as $P_1^{\vec{v}}(\sigma, a_1, a_2,\ldots, a_{\nu})$.  After the delivery of the messages, the following actions take place:

\begin{algorithmic}
	\For {$i=1$ to $i=\nu-1$}
	\If {$a_{i+1} > a_{i}$}
		For every server node $n$ in {$\{a_{i}+1, a_{i}+2,\ldots, a_{i+1}\}$} the channels from every writer in $\{C_1, C_2,\ldots,C_{\nu}\} - \{C_{\sigma(1)}, C_{\sigma(2)},\ldots, C_{\sigma(i)}\}$ to server $n$ deliver all their messages. We denote the point after the delivery of all the messages as $P_{i+1}^{\vec{v}}(\sigma, a_1, a_2,\ldots,a_{\nu})$.
		\EndIf
	\EndFor
\end{algorithmic}
Based on the above procedure, note that a server $n$ that belongs to $\{a_{i}+1, a_{i}+2,\ldots a_{i+1}\}$ receives value-dependent messages from all clients except the clients in $C_{\sigma(1)}, C_{\sigma(2)},\ldots, C_{\sigma(i)}$. For $i \in \{0,1,2,\ldots,\nu\}$, the portion of the execution $\alpha^{\vec{v}}(\sigma, a_1, a_2,\ldots,a_{\nu})$ from the initial point until point $P_{i}^{\vec{v}}(\sigma, a_1, a_2,\ldots,a_{\nu})$ is denoted as $\alpha_{i}^{\vec{v}}(\sigma, a_1, a_2,\ldots,a_{\nu}).$

The following property is due to Lemma \ref{lem:schedule0} and the fact that the only client-to-server channels act after point $P_0^{\vec{v}}$.

\begin{lemma}
	For two value vectors $\vec{v},\vec{v}'$,  for any two permutations $\sigma, \overline{\sigma}$ and integers $0\leq a_1 \leq a_2 \leq \ldots a_{\nu} \leq N-f+\nu-1$ and $0\leq \overline{a}_1 \leq \overline{a}_2 \leq \ldots \overline{a}_{\nu} \leq N-f+\nu-1$ and integers $0 \leq i_0, j_0 \leq \nu$, the state of every channel from a server to a server or a client, and the metadata components of the state of every writer in $\{C_1, C_2,\ldots,C_{\nu}\}$ is the same at point 
	$$P_{i_0}^{\vec{v}}(\sigma, a_1, a_2,\ldots,a_{\nu})$$ in execution $\alpha^{\vec{v}}(\sigma, a_1, a_2,\ldots,a_{\nu})$, as at point $$P_{j_0}^{\vec{v}'}(\overline{\sigma}, \overline{a}_1, \overline{a}_2,\ldots,\overline{a}_{\nu})$$ in execution $\alpha^{\vec{v}'}(\overline{\sigma}, \overline{a}_1, \overline{a}_2,\ldots,\overline{a}_{\nu})$.

\label{lem:metadata}
\end{lemma}
\begin{proof}
Consider a component $c$ which is a channel between two servers, or a channel from a server to a client, or a write client. The component $c$ has the same state at point   
$$P_{i_0}^{\vec{v}}(\sigma, a_1, a_2,\ldots,a_{\nu})$$ in execution $\alpha^{\vec{v}}(\sigma, a_1, a_2,\ldots,a_{\nu})$ as at point $P_0^{\vec{v}}$ of the execution, since it does not take any steps between $P_0^{\vec{v}}$ and $P_{i_0}^{\vec{v}}(\sigma, a_1, a_2,\ldots,a_{\nu}).$   Similarly, the component $c$ has the same state at point 
$$P_{j_0}^{\vec{v}'}(\overline{\sigma}, \overline{a}_1, \overline{a}_2,\ldots,\overline{a}_{\nu})$$ 
as at point $P_0^{\vec{v}'}$ in execution $\alpha^{\vec{v}'}(\overline{\sigma}, \overline{a}_1, \overline{a}_2,\ldots,\overline{a}_{\nu})$.

Lemma \ref{lem:schedule0} implies that if $c$ is a channel between servers, or a channel from a server to a client, then it has the same state at point $P_0^{\vec{v}}$ in execution $\alpha^{\vec{v}}(\sigma, a_1, a_2,\ldots,a_{\nu})$ as at point $P_0^{\vec{v}'}$ in execution $\alpha^{\vec{v}'}(\overline{\sigma}, \overline{a}_1, \overline{a}_2,\ldots,\overline{a}_{\nu}).$ This implies that $c$ has the same state at point 
	$$P_{i_0}^{\vec{v}}(\sigma, a_1, a_2,\ldots,a_{\nu})$$ in execution $\alpha^{\vec{v}}(\sigma, a_1, a_2,\ldots,a_{\nu})$, as at point $$P_{j_0}^{\vec{v}'}(\overline{\sigma}, \overline{a}_1, \overline{a}_2,\ldots,\overline{a}_{\nu})$$ in execution $\alpha^{\vec{v}'}(\overline{\sigma}, \overline{a}_1, \overline{a}_2,\ldots,\overline{a}_{\nu})$.

Similarly, Lemma \ref{lem:schedule0} implies that if $c$ is a client, then its metadata component is the same at point 
	$$P_{i_0}^{\vec{v}}(\sigma, a_1, a_2,\ldots,a_{\nu})$$ in execution $\alpha^{\vec{v}}(\sigma, a_1, a_2,\ldots,a_{\nu})$, as at point $$P_{j_0}^{\vec{v}'}(\overline{\sigma}, \overline{a}_1, \overline{a}_2,\ldots,\overline{a}_{\nu})$$ in execution $\alpha^{\vec{v}'}(\overline{\sigma}, \overline{a}_1, \overline{a}_2,\ldots,\overline{a}_{\nu})$. This completes the proof.
\end{proof}

\subsubsection{Definition of Valency}
\label{sec:thm4_valence}
Consider a collection of distinct values $v_0, v_1,\ldots, v_{\nu} \in \mathcal{V},$ a permutation $\sigma$ and a collection of numbers $a_1, a_2, \ldots, a_{\nu} \in \{0,1,2,\ldots,N-f+\nu-1\},$ where $0 \leq a_1 \leq a_2 \ldots \leq a_{\nu}\leq N-f+\nu-1.$
For a subset of write clients $\mathcal{C}_{0}\subseteq \{C_1, C_2,\ldots, C_{\nu}\}$,
and $1 \le j \le \nu$,
the point {P} 
in execution $\alpha^{\vec{v}}(\sigma, a_1, a_2,\ldots,a_{\nu})$ is said to be $(j,\mathcal{C}_{0})$-valent if there exists some execution $\beta$ which is an extension of $\alpha_i^{\vec{v}}(\sigma, a_1, a_2,\ldots,a_{\nu})$ such that, after 
$P$ in $\beta$,
\begin{itemize}
\item the writers in $\mathcal{C}_{w}-\mathcal{C}_{0}$ do not send any value-dependent messages, the channels from the writers in $\mathcal{C}_{w}-\mathcal{C}_{0}$ do not deliver any value-dependent messages, and 
\item there is a read operation that begins at a reader and completes returning value $v_j$. 
\end{itemize} 

A $(j,\{\})$-valent point is simply referred to as a $j$-valent point. It is instructive to note that a point that is $j$-valent is also $(j, \mathcal{C}_{0})$-valent for any subset $\mathcal{C}_{0} \subseteq \{C_1, C_2,\ldots, C_{\nu}\}$. 

\emph{Intuition for the definition of valency:} Before proceeding, we provide an intuitive explanation of the notion of valency. Consider a point  $P_{i}^{\vec{v}}(\sigma, a_1, a_2,\ldots,a_{\nu})$ which is $(1,C_2)$-valent. Intuitively, at such a point, the servers already have ``sufficient information'' to retrieve value $v_1.$ However, to recover $v_1$, value-dependent actions from client $C_2$ or the channels from $C_2$ could be necessary. To see this,  consider the following scenario: in the execution some algorithm, at point $P$, every server stores $v_1 + v_2,$ where the set $\mathcal{V}$ is interpreted to be some finite field, and $+$ indicates the addition operator over the field. In this case, in general, neither value $v_1$ nor value $v_2$ is retrievable from the system. However, \emph{given value $v_2$}, it can be subtracted off and the value $v_1$ would be retrievable. A clever protocol could ensure that the value-dependent messages of client $C_{2}$ will be used to subtract $v_2$ from a sufficient number of servers to ensure that $v_1$ is returnable, even if client $C_1$ did not take any value-dependent actions. In this case the point  $P_{i}^{\vec{v}}(\sigma, a_1, a_2,\ldots,a_{\nu})$ would be considered $(1,C_2)$-valent.

\subsubsection{A Key Lemma}
The following lemma is a key component of our proof of Theorem \ref{thm:fourth}.
\label{sec:thm4_lemma}
\begin{lemma} \label{lem2}
Let $\prec$ be a total ordering on $\mathcal{V}$. Given a collection of distinct values $v_1, v_2,\ldots v_{\nu} \in \mathcal{V}$, there is a permutation $\sigma:\{1,2,\ldots,\nu\}\rightarrow \{1,2,\ldots,\nu\}$, and a collection of distinct numbers $a_1, a_2, \ldots, a_{\nu} \in \{1,2,\ldots,N-f+\nu-1\},$ where $0 < a_1 < a_2 \ldots < a_{\nu} \leq N-f+\nu-1,$ such that for every $i \in \{1,2,\ldots,\nu\}$,
\begin{enumerate}[(i)]
\item $P_i^{\vec{v}}(\sigma,a_1-1, a_2-1,\dots,a_{i-1}-1, a_i, a_{i+1},\ldots,a_{\nu})$ is $(\sigma(i), \mathcal{C}_{w} - \{C_{\sigma(1)}, C_{\sigma(2)}, \ldots, C_{\sigma(i)}\})$-valent;
\item  $P_i^{\vec{v}}(\sigma,a_1-1, a_2-1,\dots,a_{i-1}-1, a_i, a_{i+1},\ldots, a_{\nu})$ is not $(\sigma(j), \mathcal{C}_{w} - \{C_{\sigma(1)}, C_{\sigma(2)}, \ldots, C_{\sigma(i)}\})$-valent, for any $1 \le j<i$; 
\item if $P_i^{\vec{v}}(\sigma,a_1-1, a_2-1,\dots,a_{i-1}-1, a_i,a_{i+1},\ldots,a_{\nu})$ is $(\sigma(j), \mathcal{C}_{w} - \{C_{\sigma(1)}, C_{\sigma(2)}, \ldots, C_{\sigma(i-1)},C_{\sigma(j)}\})$-valent for some $i < j \le \nu$, then $v_{\sigma(i)} \prec v_{\sigma(j)}$.
\end{enumerate}
\end{lemma}

In our proof of Lemma \ref{lem2}, we use Lemmas \ref{lem3}, \ref{lem6}, \ref{lem4} and \ref{lem5} which are stated and proved below. The next two lemmas state properties of the point  $P_1^{\vec{v}}(\sigma, a_1, a_2,\ldots, a_{\nu})$.
\begin{lemma} \label{lem3}
	If $a_1=N-f$, then, for every permutation $\sigma$ and integers $a_2, a_3,\ldots,a_\nu \in\{N-f, N-f+1, \ldots N-f+\nu-1\}$ where $a_1 \leq a_2 \leq \ldots a_{\nu}$,  the point $P_1^{\vec{v}}(\sigma, a_1, a_2,\ldots, a_{\nu})$ is $(i, \mathcal{C}_{w} - \{C_i\})$-valent for some $i \in \{1,2,\ldots,\nu\}$.
\end{lemma}

\begin{proof}
Let $a_1=N-f$. Consider an execution $\beta$ that extends $\alpha_{1}^{\vec{v}}(\sigma, a_1, a_2,\ldots, a_{\nu})$ as follows. Note that at point $P_{1}^{\vec{v}}(\sigma, a_1, a_2,\ldots, a_{\nu}),$ all the value-dependent messages from the clients have been delivered to the first $N-f$ servers. Therefore, for every client $C_i, i \in \{1,2,\ldots,\nu\}$, its write operation $\pi_i$ has sent its value-dependent messages, and the channels from the client to the servers have delivered all the value-dependent messages in the execution; all the send actions enabled on behalf of operation $\pi_i$ are value-independent actions. In the remainder of the execution $\beta$, the last $f$ servers do not take any steps. The clients $C_1, C_2,\ldots, C_{\nu},$ their channels and the first $N-f$ servers perform their actions in a fair manner. Since algorithm $A$ ensures that every write operation terminates in a fair execution so long as the number of server failures is no bigger than $f$ and the number of active write clients is no bigger than $\nu$, we know that a write operation $\pi_j$ completes in $\beta$ for some $j \in \{1,2,\ldots,\nu\}$. After the completion of a write operation $\pi_j$, the write clients and the channels from the write clients stop performing actions. A read operation $\pi_{r}$ begins at a reader. The reader and the first $N-f$ servers perform actions in a fair manner until read operation $\pi_r$ terminates. Because read operation $\pi_r$ begins after the termination of $\pi_{j},$ and because the algorithm satisfies weak regularity, the operation returns $v_i$ for some $i \in \{1,2,\ldots, \nu\}$. The execution $\beta$ finishes after the termination of read $\pi_r$. Thus we have created an execution $\beta,$ which is an extension of  $\alpha_{1}^{\vec{v}}(\sigma, a_1, a_2,\ldots, a_{\nu})$ such that all the clients and their channels only take value-independent output actions, and a read operation returns $v_i$. Therefore the point $P_{1}^{\vec{v}}(\sigma, a_1, a_2,\ldots, a_{\nu})$ is $i$-valent, and thus is also $(i, \mathcal{C}_{w} - \{C_i\})$-valent.
\end{proof}

\begin{lemma} \label{lem6}
	If point $P_1^{\vec{v}}(\sigma,a_1, a_2,\ldots, a_{\nu})$ is $(j,\mathcal{C}_{w}-\{C_j\})$-valent for some permutation $\sigma$  and positive integers $j, a_1, a_2,\ldots, a_{\nu}$ where $0 \leq a_1\leq a_2 \leq \ldots \leq a_{\nu}\leq N-f+\nu-1,$ then $a_1\geq 1.$
\end{lemma}
\begin{proof}

	To show that $a_{1} \geq 1$, assume to the contrary that $a_{1} = 0$. Because the point $P_1^{\vec{v}}(\sigma, a_1, a_2,\ldots, a_{\nu})$ is $(j, \mathcal{C}_{w}--\{C_j\})$-valent, there is an execution $\beta$ which extends $\alpha_{1}^{\vec{v}}(\sigma, a_1, a_2,\ldots,a_\nu)$ such that, client $C_j$ stops acting at point $P_{1}^{\vec{v}}(\sigma, a_1, a_2,\ldots,a_\nu),$ and a read operation $\pi_r$ begins and returns $v_j$. Now, because $a_1=0,$ we note that point $P_{1}^{\vec{v}}(\sigma, a_1, a_2,\ldots,a_\nu)$ is the same as point $P_0^{\vec{v}}$.   
	
	Consider the value vector $\vec{u} = (u_0, u_1, u_2,\ldots, u_{\nu}),$ where for $i \in \{0,1,\ldots,\nu\}-\{j\},$ we have $u_i = v_i,$ and $v_j \notin \{u_0, u_1,\ldots, u_\nu\}$. Note that at point $P_0^{\vec{v}},$ the state of every server, channel from a server to a server or a client, and client in $\mathcal{C}_{w}-\{C_j\}$ is the same as its state at point $P_0^{\vec{u}}$. Consider an execution $\beta'$ which extends $\alpha_0^{\vec{u}}$ as follows. Starting from point $P_0^{\vec{u}},$ every component takes the same steps as the component takes starting from point $P_0^{\vec{v}}$ in $\beta$. Note that client $C_j$ takes only value-independent actions after $P_0^{\vec{v}}$ in $\beta.$ Because Lemma \ref{lem:metadata} implies that the state of every component except client $C_j,$ and the metadata component of $C_j$ is the at point $P_0^{\vec{u}}$ in $\beta'$ as at point $P_0^{\vec{v}}$ in $\beta$, execution $\beta'$ is an execution of the algorithm. In execution $\beta'$, read operation $\pi_r$ returns $v_j$ which does not belong to $\{u_0,u_1,\ldots, u_\nu\}$. Therefore execution $\beta'$ violates weak regularity. Therefore $a_1 \geq 1$.
\end{proof}

The next lemma shows that, informally, the valency of the point $P_i^{\vec{v}}(\sigma,a_1, a_2,\ldots, a_{\nu})$ does not depend on value $a_{i+1},\dots,a_\nu$, $\sigma(i+1),\dots,\sigma(\nu)$.
\begin{lemma} \label{lem4}
	Suppose that a point $P_i^{\vec{v}}(\sigma,a_1, a_2,\ldots, a_{\nu})$ is $(j,\mathcal{C}_{w}-\{C_{\sigma(1)},C_{\sigma(2)},\ldots,C_{\sigma(i)}\})$-valent for some permutation $\sigma$ and integers $j,a_1,\ldots, a_{\nu}$ such that $0 \leq a_1 \leq a_2 \ldots \leq a_{\nu} \leq N-f+\nu-1$. Then the point $P_i^{\vec{v}}(\overline{\sigma}, a_1, a_2, \ldots, a_i, \overline{a}_{i+1}, \overline{a}_{i+2},\ldots, \overline{a}_{\nu})$ is $(j,\mathcal{C}_{w}-\{C_{\sigma(1)},C_{\sigma(2)},\ldots,C_{\sigma(i)}\})$-valent for every set of integers $\overline{a}_{i+1}, \overline{a}_{i+2},\ldots \overline{a}_{\nu}$ where $a_i \leq \overline{a}_{i+1} \leq \overline{a}_{i+2} \leq \ldots \leq \overline{a}_{\nu} \leq N-f+\nu-1$, and every $\overline{\sigma}$ where $\overline{\sigma}(l)=\sigma(l), 1 \le l \le i$.
\end{lemma}

\begin{proof}[Proof of Lemma \ref{lem4}]
	Let $P_i^{\vec{v}}(\sigma,a_1, a_2,\ldots, a_{\nu})$ be $(j,\mathcal{C}_{w}-\{C_{\sigma(1)},C_{\sigma(2)},\ldots,C_{\sigma(i)}\})$-valent. Therefore, there exists an execution $\beta$ that extends $\alpha_i^{\vec{v}}(\sigma,a_1, a_2,\ldots, a_{\nu})$ such that, after point $P_i^{\vec{v}}(\sigma,a_1, a_2,\ldots, a_{\nu})$ the clients and the channels from the clients in $C_{\sigma(1)},C_{\sigma(2)},\ldots,C_{\sigma(i)}$ do not take any value-dependent actions, and there is a read operation $\pi_r$ that begins and returns value $v_j$. Now we construct execution $\beta'$ which extends $$\alpha_i^{\vec{v}}(\overline{\sigma}, a_1, a_2, \ldots, a_i, \overline{a}_{i+1}, \overline{a}_{i+2},\ldots, \overline{a}_{\nu})$$ as follows. Note that at point $P_i^{\vec{v}}(\overline{\sigma}, a_1, a_2, \ldots, a_i, \overline{a}_{i+1}, \overline{a}_{i+2},\ldots, \overline{a}_{\nu})$ every server, channel and client is at the same state as it is at point $P_i^{\vec{v}}(\sigma,a_1, a_2,\ldots, a_{\nu})$ in $\alpha_i^{\vec{v}}(\sigma,a_1, a_2,\ldots, a_{\nu}).$ In execution $\beta'$, every component performs the same actions as the component does in $\beta$ starting from point $P_i^{\vec{v}}(\sigma,a_1, a_2,\ldots, a_{\nu}).$ Therefore, after point  $P_i^{\vec{v}}(\overline{\sigma}, a_1, a_2, \ldots, a_i, \overline{a}_{i+1}, \overline{a}_{i+2},\ldots, \overline{a}_{\nu})$
the clients and the channels from the clients in $C_{\sigma(1)},C_{\sigma(2)},\ldots,C_{\sigma(i)}$ do not take any value-dependent actions, and there is a read operation $\pi_r$ that begins and returns value $v_j$. Therefore the point $$P_i^{\vec{v}}(\overline{\sigma}, a_1, a_2, \ldots, a_i, \overline{a}_{i+1}, \overline{a}_{i+2},\ldots, \overline{a}_{\nu})$$
is $(j,\mathcal{C}_{w}-\{C_{\sigma(1)},C_{\sigma(2)},\ldots,C_{\sigma(i)}\})$-valent.
\end{proof}

The next lemma, informally, shows that if a point $P_i^{\vec{v}}(\sigma,a_1, a_2,\ldots, a_{\nu})$ has sufficient information of a value $v_j$, then an earlier point $P_k^{\vec{v}}(\sigma,a_1, a_2,\ldots, a_{\nu})$, $k \le i$, already has sufficient information of $v_j$ if we allow some extra clients to take value-dependent actions.

\begin{lemma} \label{lem5}
    Let $1 \le k \le l \le i \le \nu$.	
	If a point $P_i^{\vec{v}}(\sigma,a_1, a_2,\ldots, a_{\nu})$ is $(j,\mathcal{C}_{w}-\{C_{\sigma(1)},C_{\sigma(2)},\ldots,C_{\sigma(l)}\})$-valent for some permutation $\sigma$ and integers $j,a_1,\ldots, a_{\nu}$ such that $0 \leq a_1 \leq a_2 \ldots \leq a_{\nu} \leq N-f+\nu-1$, then for every $k \leq i$, the point $P_k^{\vec{v}}(\sigma,a_1, a_2,\ldots, a_{\nu})$ is $(j,\mathcal{C}_{w}-\{C_{\sigma(1)},C_{\sigma(2)},\ldots,C_{\sigma(k)}\})$-valent.
\end{lemma}
\begin{proof}
	Let $P_i^{\vec{v}}(\sigma,a_1, a_2,\ldots, a_{\nu})$ be $(j,\mathcal{C}_{w}-\{C_\sigma(1),C_\sigma(2),\ldots,C_{\sigma(l)}\})$-valent. Therefore, there exists an execution $\beta$ that extends $\alpha_i^{\vec{v}}(\sigma,a_1, a_2,\ldots, a_{\nu})$ such that, after point $P_i^{\vec{v}}(\sigma,a_1, a_2,\ldots, a_{\nu})$ the clients and the channels from the clients in $C_{\sigma(1)},C_{\sigma(2)},\ldots,C_{\sigma(l)}$ do not take any value-dependent actions, and there is a read operation $\pi_r$ that begins and returns value $v_j$. Note that for $k\leq i$, execution $\beta$ is an extension of $\alpha_k^{\vec{v}}(\sigma,a_1, a_2,\ldots, a_{\nu}),$ where, after point $P_k^{\vec{v}}(\sigma,a_1, a_2,\ldots, a_{\nu}),$ the clients in $\{C_\sigma(1),C_\sigma(2),\ldots, C_\sigma(k)\}$ and the channels from these clients do not perform value-dependent actions, read operation $\pi_r$ that begins and returns value $v_j.$ Therefore the point $P_k^{\vec{v}}(\sigma,a_1, a_2,\ldots, a_{\nu})$ is $(j,\mathcal{C}_{w}-\{C_{\sigma(1)},C_{\sigma(2)},\ldots,C_{\sigma(k)}\})$-valent.
\end{proof}

We are now ready to prove our key lemma, Lemma \ref{lem2}.
\begin{proof}[Proof of Lemma \ref{lem2}]
We begin by choosing $a_1, \sigma(1)$.	From Lemma \ref{lem3}, we know that the set
\begin{equation}\mathcal{A}_1 = \left\{(\overline{a}_{1},i): 
	\begin{array}{c}\textrm{there exists an integer $i \in \{1,2,\ldots,\nu\}$, a permutation } \overline{\sigma}, \textrm{and } \\ \textrm{integers }\overline{a}_{1},\ldots \overline{a}_{\nu}$\textrm{ where } $0 \leq \overline{a}_{1} \leq \overline{a}_{2} \leq \ldots \leq \overline{a}_{\nu} \leq N-f+\nu-1 \\ \textrm {such that }P_1^{\vec{v}}(\overline{\sigma}, \overline{a}_{1}, \overline{a}_2,\ldots, \overline{a}_{\nu}) \\ \textrm{is $(i, \mathcal{C}_{w}-\{C_i\})$-valent}\end{array}\right\} 
\end{equation}
is non-empty. In particular, Lemma \ref{lem3} implies that tuple $(N-f,i)$ belongs to $\mathcal{A}_1$ for some integer $i\in\{1,2,\ldots,\nu\}$. We choose $a_1$ to be the smallest integer such that $(a_1, i)$ belongs to the set $\mathcal{A}_{1}$ for some $i$\footnote{Informally, $a_1$ may be viewed as the smallest number such that the first $a_1$ servers contains ``sufficient information'' of some value $v_i,$ given the information of all other values.}, that is \[a_{1} = \min\{a:\textrm{there exists } j\textrm{ such that }(a,j) \in \mathcal{A}_{1}\}.\] Note that Lemma \ref{lem6} implies that $a_1 \geq 1$.
{We let 
\begin{align*}
\sigma(1) = \arg\min_{\{j:(a_{1},j) \in \mathcal{A}_{1}\}}v_j ,
\end{align*}
where the minimum is according to the ordering $\prec$.}

Next, we provide a procedure that recursively chooses $a_{i_0+1}$ and $\sigma(i_0+1)$  given $a_{1},a_{2},\ldots, a_{i_0},$ and $\sigma(1),\sigma(2),\ldots,\sigma(i_0),$ for any $i_0 \in \{1,2,\ldots,\nu-1\}$. Our choice of $a_{i_0+1}$ will satisfy $a_{i_0} < a_{i_0+1} \leq N-f+i_0-1$.

Let 
\begin{equation}\label{eqA}\mathcal{A}_{i_0+1} = \left\{(\overline{a}_{i_0+1},i): 
	\begin{array}{c}\textrm{there exists a permutation } \overline{\sigma}, \textrm{where $\overline{\sigma}(j) = \sigma(j)$ for $j\leq i_0$, and} \\ \textrm{integers }\overline{a}_{i_0+1}, \overline{a}_{i_0+2},\ldots \overline{a}_{\nu}$\textrm{ where } $a_{i_0} \leq \overline{a}_{i_0+1} \leq \overline{a}_{i_0+2} \leq \ldots \leq \overline{a}_{\nu} \leq N-f+\nu-1 \\ \textrm{ and an integer $i \in \{1,2\ldots,\nu\}-\{\sigma(1),\sigma(2),\ldots, \sigma(i_0)\}$} \\ \textrm {such that }P_{i_0+1}^{\vec{v}}(\overline{\sigma}, {a}_{1}-1, a_2-1, \ldots, a_{i_0}-1, \overline{a}_{i_0+1}, \overline{a}_{i_0+2},\ldots, \overline{a}_{\nu}) \\ \textrm{is $(i, \mathcal{C}_{w}-\{C_{\sigma(1)}, C_{\sigma(2)},\ldots, C_{\sigma(i_0)}, C_i\})$-valent} \end{array}\right\}
\end{equation}

We show that $\mathcal{A}_{i_0+1}$ is non-empty. We choose $a_{i_0+1}$ to be the smallest integer such that $(a_{i_0+1}, i)$ belongs to the set $\mathcal{A}_{i_0+1}$ for some $i$. 
{We let
\begin{align*}
\sigma(i_0+1) = \displaystyle\arg\displaystyle\min_{\{j:(a_{i_0+1},j) \in \mathcal{A}_{i_0+1}\}}v_j,
\end{align*}
where the minimum above is taken according to the ordering $\prec$.
}

To complete the proof of Lemma \ref{lem2}, we show that 
\begin{enumerate}[(a)]
	\item $\mathcal{A}_{i_0+1}$ is non-empty, {and $a_{i_0+1}>a_{i_0}$, $1 \le i_0 \le \nu-1$;}
\item $P_{i_0+1}^{\vec{v}}(\sigma, a_1-1, a_2-1,\dots,a_{i_0}-1, a_{i_0+1}, a_{i_0+2},\ldots,a_{\nu})$ is $(\sigma(i_0+1), \mathcal{C}_{w} - \{C_{\sigma(1)}, C_{\sigma(2)}, \ldots, C_{\sigma(i_0+1)}\})$-valent, $0 \le i_0 \le \nu-1$;
\item  $P_{i_0+1}^{\vec{v}}(\sigma, a_1-1, a_2-1,\dots,a_{i_0}-1, a_{i_0+1}, a_{i_0+2},\ldots,a_{\nu})$ is not $(\sigma(j), \mathcal{C}_{w} - \{C_{\sigma(1)}, C_{\sigma(2)}, \ldots, C_{\sigma(i_0+1)}\})$-valent, for any $j<i_0+1$, $1 \le i_0 \le \nu-1$; 
\item if $P_{i_0+1}^{\vec{v}}(\sigma, a_1-1, a_2-1,\dots,a_{i_0}-1, a_{i_0+1},a_{i_0+2},\ldots,a_{\nu})$ is $$(\sigma(j), \mathcal{C}_{w} - \{C_{\sigma(1)}, \ldots, C_{\sigma(i_0)},C_{\sigma(j)}\})\textrm{-valent}$$for some $j>i_0+1$, then $v_{\sigma(i_0+1)} \prec v_{\sigma(j)}$, $0 \le i_0 \le \nu-2$.
\end{enumerate}

\emph{Proof of $(a)$:}

We show $(a)$ by showing that there is some integer $i$ such that $(N-f+i_0,i)$ belongs to $\mathcal{A}_{i_0+1}$. More specifically, we will show that for any arbitrary permutation $\overline{\sigma}$ which satisfies $\overline{\sigma}(j) = \sigma(j)$ for $1 \leq j \leq i_0,$ the point $P_{i_0+1}^{\vec{v}}(\overline{\sigma}, {a}_{1}-1, a_2-1, \ldots, a_{i_0}-1, N-f+i_0, N-f+i_0+1,\ldots, N-f+\nu-1)$ is $(i, \mathcal{C}_{w}-\{C_{\sigma(1)}, C_{\sigma(2)},\ldots, C_{\sigma(i_0)}, C_i\})$-valent for some $i \in \{1,2,\ldots,\nu\}-\{\sigma(1),\sigma(2),\ldots,\sigma(i_0)\}$. To show this, we construct an execution $\beta,$ which is an extension of $\alpha_{i_0+1}^{\vec{v}}(\overline{\sigma}, {a}_{1}-1, a_2-1, \ldots, a_{i_0}-1, N-f+i_0, N-f+i_0+1,\ldots, N-f+\nu-1)$  as follows. After point $P_{i_0+1}^{\vec{v}}(\overline{\sigma}, {a}_{1}-1, a_2-1, \ldots, a_{i_0}-1, N-f+i_0, N-f+i_0+1,\ldots, N-f+\nu-1)$ in $\beta$, the write clients in $\{C_{\overline{\sigma}(1)}, C_{\overline{\sigma}(2)},\ldots, C_{\overline{\sigma}(i_0)}\}$, and the channels from these write clients to the non-failed servers do not send or deliver value-dependent messages. The clients in $\mathcal{C}_{w}-\{C_{\overline{\sigma}(1)}, C_{\overline{\sigma}(2)},\ldots, C_{\overline{\sigma}(i_0)}\}$, the channels from these clients, the non-failed servers, and the channels between the servers continue taking actions in a fair manner.  Note that algorithm A guarantees that in a fair execution where the number of server failures is no bigger than $f$ and the number of active write clients is no bigger than $\nu$, every write operation terminates. From the perspective of clients in $\mathcal{C}_{w}-\{C_{\overline{\sigma}(1)}, C_{\overline{\sigma}(2)},\ldots, C_{\overline{\sigma}(i_0)}\},$ the execution $\beta$ is indistinguishable from a fair execution. Therefore some operation in $\{\pi_{\overline{\sigma}(i_0+1)}, \pi_{\overline{\sigma}(i_0+2)},\ldots,\pi_{\overline{\sigma}(\nu)}\}$ terminates in $\beta$. After the termination of an operation $\pi_j, j \in \{\overline{\sigma}(i_0+1),\overline{\sigma}(i_0+2),\ldots,\overline{\sigma}(\nu)\},$ all the write operations and their channels stop taking actions. A read operation $\pi_{r}$ begins at a read client. The reader, the channels from and to the reader, and the non-failed servers perform actions in a fair manner until the read operation $\pi_r$ terminates. After the termination of $\pi_r$ the execution $\beta$ ends.
Because the algorithm satisfies regularity and because write operation $\pi_j$ has terminated, the read operation $\pi_r$ returns value $v_i$ for some $i \in \{1,2,\ldots,\nu\}$. 

We show that, in fact, $\pi_r$ returns $v_i$ for some $i \in \{1,2,\ldots,\nu\}-\{\sigma(1),\sigma(2),\ldots,\sigma(i_0)\}.$ Assume the contrary, that is, assume that the read operation $\pi_r$ returns $v_{\sigma(k)}$ for $k \in \{1, 2,\ldots, i_0\}.$ The existence of execution $\beta$ implies that the point $P_{i_0+1}^{\vec{v}}(\overline{\sigma}, {a}_{1}-1, a_2-1, \ldots, a_{i_0}-1, N-f+i_0, N-f+i_0+1,\ldots, N-f+\nu-1)$ is $(\sigma(k),\mathcal{C}_{w}-\{C_{{\sigma}(1)}, C_{{\sigma}(2)},\ldots, C_{{\sigma}(i_0)}\})$-valent. {Lemma \ref{lem5}} implies that $P_{k}^{\vec{v}}(\overline{\sigma}, {a}_{1}-1, a_2-1, \ldots, a_{i_0}-1, N-f+i_0, N-f+i_0+1,\ldots, N-f+\nu-1)$ is $(\sigma(k),\mathcal{C}_{w}-\{C_{{\sigma}(1)}, C_{{\sigma}(2)},\ldots, C_{{\sigma}(k)}\})$-valent. Therefore $(a_{k}-1, \sigma(k)) \in \mathcal{A}_{k}$. This however contradicts the fact that we choose $a_k$ to be the smallest element such that $(a_k,j)$ is $\mathcal{A}_{k}$ for some $j \in \{1,2,\ldots, \nu\}-\{\sigma(1),\sigma(2),\ldots,\sigma(k-1)\}$. Therefore, it cannot be that $\pi_r$ returns $v_{\sigma(k)}$ for some $k \in\{1,2,\ldots,i_0\}$. Therefore, $\mathcal{A}_{i_0+1}$ is non-empty.

{We now show that $a_{i_0+1}> a_{i_0}$. We know $a_{i_0+1} \ge a_{i_0}$.
To show that $a_{i_0+1} > a_{i_0}$, assume to the contrary that $a_{i_0+1} = a_{i_0}$. Because $a_{i_0+1}\in \mathcal{A}_{i_0+1}$ and because we assume $a_{i_0+1}=a_{i_0}$, point $P_{i_0+1}^{\vec{v}}({\overline{\sigma}},a_1-1,\dots,a_{i_0}-1,a_{i_0},\overline{a}_{i_0+1},\dots,\overline{a}_\nu)$ is $(\sigma(i_0+1), \mathcal{C}_{w}-\{C_{\sigma(1)}, C_{\sigma(2)},\ldots, C_{\sigma(i_0)}, C_{\sigma(i_0+1)}\})$-valent. By Lemma \ref{lem5} this implies that point $P_{i_0}^{\vec{v}}({\overline{\sigma}},a_1-1,\dots,a_{i_0}-1,a_{i_0},\overline{a}_{i_0+1},\dots,\overline{a}_\nu)$ is $(\sigma(i_0+1), \mathcal{C}_{w}-\{C_{\sigma(1)}, C_{\sigma(2)},\ldots, C_{\sigma(i_0)}\})$-valent. Therefore, $(a_{i_0}-1,\sigma(i_0+1)) \in \mathcal{A}_{i_0}$, and contradicts the fact that we choose $a_{i_0}$ to be the smallest element such that $(a_{i_0},j)$ is in $\mathcal{A}_{i_0}$ for some $j \in \{1,2,\ldots, \nu\}-\{\sigma(1),\sigma(2),\ldots,\sigma(i_0-1)\}$.
}

\emph{Proof of (b): }

Because $(a_{i_0}, \sigma(i_0)) \in \mathcal{A}_{i_0}, 1 \le i_0 \le \nu$ there are distinct integers $\overline{a}_{i_0+1},\overline{a}_{i_0+2},\ldots, \overline{a}_{\nu},$ and a permutation $\overline{\sigma}$ such that 
\begin{itemize}
\item $\overline{\sigma}(j) = \sigma(j)$ for $1 \leq j \leq i_0$
\item $a_{i_0} \le \overline{a}_{i_0+1} \le \overline{a}_{i_0+2} \le  \ldots \overline{a}_{\nu} \le N-f+\nu-1$
\end{itemize}
such that the point $$P_{i_0}^{\vec{v}}(\overline{\sigma}, {a}_{1}-1, a_2-1, \ldots, a_{i_0-1}-1, a_{i_0}, \overline{a}_{i_0+1}, \overline{a}_{i_0+2},\ldots, \overline{a}_{\nu})$$ is $(\sigma(i_0), \mathcal{C}_{w}-\{C_{\sigma(1)}, C_{\sigma(2)},\ldots, C_{\sigma(i_0)}\})$-valent. Using Lemma \ref{lem4}, we conclude that the point $$P_{i_0}^{\vec{v}}(\sigma, a_1-1, a_2-1,\ldots,a_{i_0-1}-1, a_{i_0}, a_{i_0+1},\ldots,a_{\nu})$$ is $(\sigma(i_0), \mathcal{C}_{w} - \{C_{\sigma(1)}, C_{\sigma(2)}, \ldots, C_{\sigma(i_0)}\})$-valent. 

\emph{Proof of (c):}

If point $P_{i_0+1}^{\vec{v}}(\sigma, a_1-1, a_2-1,\dots,a_{i_0}-1, a_{i_0+1},a_{i_0+2},\ldots,a_{\nu})$ is $(\sigma(j), \mathcal{C}_{w} - \{C_{\sigma(1)}, C_{\sigma(2)}, \ldots, C_{\sigma(i_0+1)}\})$-valent for $j<i_0+1$, then by Lemma \ref{lem5}, the point $P_{j}^{\vec{v}}(\sigma, a_1-1, a_2-1,\dots,a_{i_0}-1, a_{i_0+1},a_{i_0+2},\ldots,a_{\nu})$ is $(\sigma(j), \mathcal{C}_{w} - \{C_{\sigma(1)}, C_{\sigma(2)}, \ldots, C_{\sigma(j)}\})$-valent. However, this implies that $(a_{j}-1, \sigma(j))\in \mathcal{A}_{j}$, which contradicts the fact that $a_{j}$ is the smallest integer among all the integers such that $(a_j, k) \in \mathcal{A}_{j}.$  Therefore $P_{i_0+1}^{\vec{v}}(\sigma, a_1-1, a_2-1,\dots,a_{i_0}-1, a_{i_0+1},a_{i_0+2},\ldots,a_{\nu})$ is not $(\sigma(j), \mathcal{C}_{w} - \{C_{\sigma(1)}, C_{\sigma(2)}, \ldots, C_{\sigma(i_0)}\})$-valent, for any $j<i_0$.

\emph{Proof of $(d)$:}
If $P_{i_0+1}^{\vec{v}}(\sigma, a_1-1, a_2-1,\dots,a_{i_0}-1, a_{i_0+1},a_{i_0+2},\ldots,a_{\nu})$ is $$(\sigma(j), \mathcal{C}_{w} - \{C_{\sigma(1)}, \ldots, C_{\sigma(i_0)},C_{\sigma(j)}\})\textrm{-valent}$$for some $j>i_0+1$,
then  $(a_{i_0+1},\sigma(j))$ belongs to $\mathcal{A}_{i_0+1}$. Because $(a_{i_0+1},\sigma(i_0+1))$ also belongs to $\mathcal{A}_{i_0+1}$ and because we chose 
$$ 
\sigma(i_0+1) = \arg\min_{\{(a_{i_0+1},k) \in \mathcal{A}_{i_0+1}\}}v_k,$$
we have $v_{\sigma(i_0+1)} \prec v_{\sigma(j)}.$  
\end{proof}

\subsubsection{Proof of Theorem \ref{thm:fourth}}
\label{sec:thm4_proof}
Recall that $\mathcal{S}_{n}$ represents the set of possible server states of the $n$th server. Let $\vec{S}_i^{\vec{v}}(\sigma, a_1, a_2,\ldots,a_{\nu})$ denote the $N-f+\nu-1$ dimensional vector in the set $\prod_{n=1}^{N-f+\nu-1}\mathcal{S}_{n}$, whose $j$th component denotes the state of server $j$ at point 
$$P_i^{\vec{v}}(\sigma, a_1, a_2,\ldots,a_{\nu})$$ in execution 
$$\alpha_i^{\vec{v}}(\sigma, a_1, a_2,\ldots,a_{\nu}).$$ 

For a given vector $\vec{v}=(v_0, v_1,\ldots,v_{\nu}) \in \mathcal{V}^{\nu},$ let permutation $\sigma^{\vec{v}}$ and distinct integers $a_1^{\vec{v}}, \ldots, a_{\nu}^{\vec{v}}$ satisfy the conditions of Lemma \ref{lem2} as per a total order $\prec$ on the set $\mathcal{V}$. 

Let $v_{0} \in \mathcal{V}$ be the initial value , and 
$$\mathcal{V}_0 = \{(v_0, v_1,\ldots, v_{\nu}):v_1, v_2,\ldots, v_{\nu}\in \mathcal{V}-\{v_0\} \textrm{ are distinct}\}.$$

We show that there is a one-to-one mapping from tuples of the form 
$$(\sigma^{\vec{v}}, a_1^{\vec{v}}, a_2^{\vec{v}},\ldots, a_{\nu}^{\vec{v}}, \vec{S}_{\nu}^{\vec{v}}(\sigma^{\vec{v}}, a_1^{\vec{v}}, a_2^{\vec{v}},\ldots,a_{\nu}^{\vec{v}}))$$
to the vectors in $\mathcal{V}_{0}.$ Specifically, if $\vec{u}, \vec{v} \in \mathcal{V}_0$ and $\vec{u}\neq \vec{v}$, we show that 
\begin{align}
	&(\sigma^{\vec{u}}, a_1^{\vec{u}}, a_2^{\vec{u}},\ldots, a_{\nu}^{\vec{u}}, \vec{S}_{\nu}^{\vec{u}}(\sigma^{\vec{u}}, a_1^{\vec{u}}, a_2^{\vec{u}},\ldots,a_{\nu}^{\vec{u}}))\nonumber \\
	& \neq (\sigma^{\vec{v}}, a_1^{\vec{v}}, a_2^{\vec{v}},\ldots, a_{\nu}^{\vec{v}}, \vec{S}_{\nu}^{\vec{v}}(\sigma^{\vec{v}}, a_1^{\vec{v}}, a_2^{\vec{v}},\ldots,a_{\nu}^{\vec{v}})). \label{eq:1}
\end{align}
The one-to-one mapping implies that the cardinality of $\mathcal{P} \times \{1,2,\ldots,N-f+\nu-1\}^{\nu}\times \prod_{n=1}^{N-f+\nu-1}\mathcal{S}_{n}$ must be no smaller than $|\mathcal{V}_0|,$ where $\mathcal{P}$ represents the set of all permutations on $\{1,2,\ldots,\nu\}.$ This implies that 
\begin{eqnarray*}
	(\nu !) \cdot (N-f+\nu-1)^{\nu} \cdot \prod_{n=1}^{N-f+\nu-1} |\mathcal{S}_{n}| &\geq &|\mathcal{V}_0| = \binom{|\cV|-1}{\nu}
\end{eqnarray*}
which implies the statement of the theorem.

To complete the theorem, we show the relation stated in (\ref{eq:1}). We provide a proof by contradiction. Consider two distinct vectors $\vec{u}=(v_0, u_1,\ldots, u_{\nu})$ and $\vec{v}=(v_0, v_1,\ldots, v_{\nu})$ which violate (\ref{eq:1}). Let $\sigma=\sigma^{\vec{u}} = \sigma^{\vec{v}}$ and $a_i = a_i^{\vec{u}}=a_i^{\vec{v}}$ for $i \in \{1,2,\ldots,\nu\}$.

Since $\vec{u} \neq \vec{v},$ we know that there exists an index $i\in\{1,2\ldots,\nu\}$ such that $u_i \neq v_i$. Let $i_{0}$ be the largest element of $\{j:u_{\sigma(j)} \neq v_{\sigma(j)}\}.$  Note that if $j > i_0,$ we have $u_{\sigma(j)} = v_{\sigma(j)}.$ Also, $u_{\sigma(i_0)} \neq v_{\sigma(i_0)}$. This implies that either $u_{\sigma(i_0)} \prec v_{\sigma(i_0)}$ or $v_{\sigma(i_0)} \prec u_{\sigma(i_0)}.$ Without loss of generality, we assume that $v_{\sigma(i_0)} \prec u_{\sigma(i_0)}.$  We next use Lemma \ref{lem2} to show that $u_{\sigma(i_0)}=v_{\sigma(i_0)}$ which is a contradiction.

Because the point $P^{\vec{v}}_{i_0}(\sigma, a_1-1, a_2-1,\ldots,a_{i_0}-1, a_{i_0}, a_{i_0+1},\ldots,a_{\nu})$ is $(\sigma(i_0), \mathcal{C}_{w}-\{C_{\sigma(1)}, C_{\sigma(1)},\ldots,C_{\sigma(i_0)}\})$-valent, there exists an execution $\beta'$ which extends $$\alpha_{i_0}^{\vec{v}}(\sigma, a_1-1, a_2-1,\ldots,a_{i_0-1}-1, a_{i_0}, a_{i_0+1},\dots,a_{\nu}),$$ such that after $P^{\vec{v}}_{i_0}(\sigma, a_1-1, a_2-1,\ldots,a_{i_0}-1, a_{i_0}, a_{i_0+1},\ldots,a_{\nu})$ the clients in $\{C_{\sigma(1)}, C_{\sigma(2)},\ldots,C_{\sigma(i_0)}\}$ and the channels from these clients do not take value-dependent actions and there is a read operation that begins and returns ${v}_{\sigma(i_0)}$. 

We compare the component states at two points: point $P^{\vec{v}}_{i_0}(\sigma, a_1-1, a_2-1,\ldots,a_{i_0}-1, a_{i_0}, a_{i_0+1},\ldots,a_{\nu})$ in execution $\beta'$ and point $P^{\vec{u}}_{i_0}(\sigma, a_1-1, a_2-1,\ldots,a_{i_0}-1, a_{i_0}, a_{i_0+1},\ldots,a_{\nu})$ in execution $\alpha_{i_0}^{\vec{u}}(\sigma, a_1-1, a_2-1,\ldots,a_{i_0-1}-1, a_{i_0}, a_{i_0+1},\ldots,a_{\nu}).$
Because $\vec{S}_\nu^{\vec{u}} = \vec{S}_\nu^{\vec{v}},$ the state of a server at the first point is the same as its state at the second point.
From Lemma \ref{lem:metadata}, the metadata components of the state of any write client and the state of every channel from a server to a server or client are the same at both points.
Finally, because $u_{\sigma(j)} = v_{\sigma(j)}$ for all $j > i_0,$ and the clients in $C_{w}-\{C_{\sigma(1)}, C_{\sigma(2)},\ldots, C_{\sigma(i_0)}\}$ have the same state at both points.

We now create an execution $\beta$ that extends $\alpha_{i_0}^{\vec{u}}(\sigma, a_1, a_2,\ldots,a_{\nu})$ as follows. Starting at point $P^{\vec{u}}_{i_0}(\sigma, a_1-1, a_2-1,\ldots,a_{i_0}-1, a_{i_0}, a_{i_0+1},\ldots,a_{\nu}),$ every client, server, and channel takes the same steps in $\beta$ as it takes starting from point  $P^{\vec{v}}_{i_0}(\sigma, a_1-1, a_2-1,\ldots,a_{i_0}-1, a_{i_0}, a_{i_0+1},\ldots,a_{\nu})$ in $\beta'.$ Note that every component except the clients in  $\{C_{\sigma(1)}, C_{\sigma(2)},\ldots, C_{\sigma(i_0)}\}$ and the channels from these clients have the same state at point $P^{\vec{u}}_{i_0}(\sigma, a_1-1, a_2-1,\ldots,a_{i_0}-1, a_{i_0}, a_{i_0+1},\ldots,a_{\nu})$  in $\beta$ as at point $P^{\vec{v}}_{i_0}(\sigma, a_1-1, a_2-1,\ldots,a_{i_0}-1, a_{i_0}, a_{i_0+1},\ldots,a_{\nu})$ in $\beta'$. The metadata components of the states of the clients in $\{C_{\sigma(1)}, C_{\sigma(2)},\ldots, C_{\sigma(i_0)}\}$ and the metadata messages in the channels from these clients are the same at the two points. Because clients in $\{C_{\sigma(1)}, C_{\sigma(2)},\ldots, C_{\sigma(i_0)}\}$  and the channels from these clients only take value-independent output actions after point $P^{\vec{v}}_{i_0}(\sigma, a_1-1, a_2-1,\ldots,a_{i_0}-1, a_{i_0}, a_{i_0+1},\ldots,a_{\nu})$ in $\beta'$,  $\beta$ is a valid execution of the algorithm $A$. Thus, we have created an execution $\beta$ which is an extension of  $\alpha^{\vec{u}}_{i_0}(\sigma, a_1-1, a_2-1,\ldots,a_{i_0}-1, a_{i_0}, a_{i_0+1},\ldots,a_{\nu}),$ where, after point $P^{\vec{v}}_{i_0}(\sigma, a_1-1, a_2-1,\ldots,a_{i_0}-1, a_{i_0}, a_{i_0+1},\ldots,a_{\nu}),$ the clients in $\{C_{\sigma(1)}, C_{\sigma(2)},\ldots, C_{\sigma(i_0)}\}$  and the channels from these clients do not take any value-dependent actions, and a read operation begins and returns $v_{\sigma(i_0)}.$ 
Because the algorithm is weakly regular, we must have $v_{\sigma(i_0)} \in \{v_0, u_1, u_2,\ldots, u_{\nu}\}$. Furthermore, $v_{\sigma(i_0)}\neq v_0$ since we assume that the components of $\vec{v}$ are distinct. So, there exists an integer $j_0$ in $\{1,2,\ldots,\nu\}$ such that $v_{\sigma(i_0)}=u_{\sigma(j_0)}.$ 

The existence of execution $\beta$ implies that the point $P^{\vec{u}}_{i_0}(\sigma, a_1-1, a_2-1,\ldots,a_{i_0}-1, a_{i_0}, a_{i_0+1},\ldots,a_{\nu}),$ is $(\sigma(j_0), \{C_{\sigma(1)}, C_{\sigma(2)},\ldots, C_{\sigma(i_0)}\})$-valent. Statement (ii) of Lemma \ref{lem2} implies that $j_0 \geq i_0$. Statement (iii) of Lemma \ref{lem2} implies that if $j_0 > i_0,$ then $u_{\sigma(i_0)} \prec u_{\sigma(j_0)}$. However, we started with the assumption that  $u_{\sigma(j_0)} = v_{\sigma(i_0)} \prec u_{\sigma(i_0)}.$ Therefore, it cannot happen that $j_0 > i_0$, and we conclude that $j_0 = i_0, v_{\sigma(i_0)} = u_{\sigma(i_0)}.$ This contradicts our assumption that $v_{\sigma(i_0)} \neq u_{\sigma(i_0)}.$ Therefore, (\ref{eq:1}) must be true. This completes the proof.

\subsection{Conjecture related to Theorem \ref{thm:fourth}}
\label{sec:fourth_discussions}
The assumptions of Theorem \ref{thm:fourth} do not apply to some algorithms \cite{HGR, androulaki2014_separate_metadata}. These algorithms send value-dependent messages in two phases. In one of the two phases, the algorithms send erasure coded elements corresponding to the value. In the other phase where value-dependent messages are sent, a hash of the value is sent. The hashes are used for verification of the client's integrity, which is important in \cite{HGR, androulaki2014_separate_metadata} as the algorithms in these references handle Byzantine failures. We believe that it may be possible to generalize our result of Theorem \ref{thm:fourth} with Assumption 3 (b) modified as follows:
\begin{itemize}
\item  the algorithm has a bounded number of phases, and 
\item  there is at most one phase where a value-dependent message of size $\Theta(|\mathcal{V}|)$ is sent.
\end{itemize}
The above restrictions imply that, even if there is more than one phase where value dependent messages are sent, the value-dependent messages in the additional phases do not carry much information about the value. The above restrictions would cover algorithms of \cite{HGR, androulaki2014_separate_metadata}, and we conjecture that the lower bound of Corollary \ref{cor:fourth} bound still applies.

\section{Concluding Remarks}
\label{sec:conclusion}
{
This paper was motivated by the following open question (see Section \ref{sec:comparisons}): Does there exist an atomic shared memory emulation algorithm whose storage cost is smaller than $\nu\frac{N}{N-f} \log_{2}|\mathcal{V}|,$ where $\nu$ represents the number of active write operations? This question remains open. The insight obtained by our bounds in conjunction with the result of \cite{spiegelman2015space} is summarized here. 
If there is an algorithm whose storage cost is $g(\nu, N, f) \log_{2}|\mathcal{V}|+o(\log_{2}|\mathcal{V}|),$  where $g(\nu,N,f)$ is some real-valued function of parameters $\nu, N, f$ then
\begin{itemize}
	\item $g(\nu,N,f) \geq \frac{2N}{N-f+2};$ 
	\item If $g(\nu,N,f) <  \frac{\nu N}{N-f+\nu-1},$ then \begin{itemize}
			\item 
			the writer sends its value in multiple phases to the servers, or 
\item the writer's state may not separate the value and the metadata, or
\item during a write operation, the writer can take non-black box actions;
	\end{itemize}

\item If, for a given values of parameters $N,f$, we have  $g(\nu,N,f) < f+1$ for all values of $\nu,$ then, in certain executions, the servers store symbols which jointly encode values across different versions.
\end{itemize}


\newpage

\bibliographystyle{abbrv}
\bibliography{biblio}

\appendix
\section{Discussion on the Storage Scheme Assumption of \cite{spiegelman2015space}}
\label{app:discussion}
For the sake of technical clarity and completeness, we provide a discussion on the storage scheme assumption of \cite{spiegelman2015space}. Reference \cite{spiegelman2015space} assumes that every stored bit is associated uniquely with a value of a write operation. However, this assumption is restrictive, and the storage cost lower bound proof of \cite{spiegelman2015space} may not be applicable for arbitrary storage schemes.

In fact, a critical idea in the proof of \cite{spiegelman2015space} is the following. If at a point $P$ in an execution of an algorithm $A$, 
\begin{itemize}
\item no value from $\{v_1, v_2,\ldots,v_{m}\}\subset \mathcal{V}$ is returnable from a set of servers,
\item if there is a value $v_i \in \{v_1, v_2,\ldots, v_m\}$ such that no server stores a single bit of the value, and 
\item after a single step of the execution, some value which is not necessarily  $\{v_1, v_2,\ldots,v_{m}\}\subset \mathcal{V}$ is returnable,
\end{itemize}  
then the number of bits stored in a server must increase by $\log_{2}|\mathcal{V}|$ bits in the step. However, this is true only for specific storage schemes, but may not be true for arbitrary storage schemes. We show this via a counter-example. 

Let $\mathcal{V}$ be a finite field of $2^{m}$ elements for some integer $m$. Note that every element of $\mathcal{V}$  is an $m$-bit vector over the binary base field. Let $v_1, v_2, v_3 \in \mathcal{V}$ be three versions of the data object associated with three write operations in an execution of an algorithm. Suppose that, in some algorithm $A$, because of structure of the server protocol, there are two servers which both store $v_1 + v_2 + v_3$ at some point $P$ of an execution. It is impossible to  associate a bit stored at the servers with a single value.  Note that a reader which accesses only the two servers cannot recover even a single bit of $v_1, v_2, v_3$.

For argument's sake, suppose that the bits stored are associated with  none of the values. Note that none of $\{v_1, v_2, v_3\}$ are returnable from both servers.  Now, imagine a single step, where the first server receives a message that contains $v_2$ and, after the receipt of the message, stores $v_1 + v_3$ at point $P'$. Then a reader which can access the bits stored in both servers can recover $v_2$ by simply subtracting the contents of the two servers. However, the number of bits stored in the servers did not change in the step!  Thus the proof method of \cite{spiegelman2015space} would not be generally applicable to algorithm $A$.

The proof technique of \cite{spiegelman2015space}, in fact, applies when versions are  encoded \emph{separately}. For instance, if a server that receives information from three values $v_1, v_2, v_3$, and stores information of the form $(f_1(v_1), f_2(v_2), f_3(v_3))$, where $f_1, f_2, f_3$ are three arbitrary functions. In this instance, every stored bit can be uniquely associated with a write operation.

It must be noted that the algorithm $A$ considered in this section is hypothetical. We do not know whether we can construct meaningful algorithms that can encode across different versions. Our storage cost lower bound of Theorem \ref{thm:fourth} suggests that even if such an algorithm can be constructed, there would be little benefit from a storage cost perspective if the writers send values in a single phase.

\section{A Simple Information-theoretic Lower Bound}
\label{sec:firstthm}

In this section, we provide a simple information-theoretic lower bound on the storage cost incurred by any distributed shared memory emulation algorithm. We describe the lower bound in Theorem \ref{thm:first}. The theorem leads to lower bounds on the max- and {total-}storage costs, which are stated in Corollary \ref{cor:first}. After stating Theorem \ref{thm:first} and Corollary \ref{cor:first}, we provide an informal description of the proof of Theorem \ref{thm:first}, followed by a formal description.

\subsection{Statement of Theorem \ref{thm:first}}

\begin{theorem}
{Let $A$ be a single-writer single-reader shared memory emulation algorithm that implements a {regular} read-write object whose values come from a finite set $\mathcal{V}$. Suppose that,  in algorithm $A$, every server's state belongs to a set $\mathcal{S}$.  Suppose that the algorithm $A$ satisfies the following liveness property: 

In a fair execution of $A,$ if the number of server failures is no bigger than $f$, $f\ge 1$, then every operation invoked at a non-failing client terminates.}

{Then, for every subset $\mathcal{N} \subset \{1,2,\ldots,N\}$ where $|\mathcal{N}| = N-f$, 
$$\sum_{n \in \mathcal{N}} \log_{2} |\mathcal{S}_n| \geq \log_{2}|\mathcal{V}|.$$}
\label{thm:first}
\end{theorem}

The above theorem naturally implies a bound on the total- and max-storage costs as demonstrated in the following corollary.

\begin{corollary}
Let $A$ be a single-writer-single-reader shared memory emulation algorithm that implements a {regular} read-write object whose values come from a finite set $\mathcal{V}$. Suppose that every server's state belongs to a set $\mathcal{S}$ in algorithm $A$.  Suppose that the algorithm $A$ satisfies the following liveness property: 

In a fair execution of $A,$ if the number of server failures is no bigger than $f$, $f\ge 1$, then every operation invoked at a non-failing client terminates.

Then
\begin{eqnarray*}
MaxStorage(A) &\geq& \frac{\log_2{|\mathcal{V}|}}{N-f}, \textrm{ and }\\
TotalStorage(A) &\geq& \frac{N \log_2{|\mathcal{V}|}}{N-f}.
\end{eqnarray*}
\label{cor:first}
\end{corollary}

\begin{proof}[Proof of Corollary \ref{cor:first}]
	We assume, without loss of generality, that $|\mathcal{S}_{1}| \leq |\mathcal{S}_{2} \leq \ldots \leq |\mathcal{S}_{N}|.$ From Theorem \ref{thm:first}, we have 
	$$\sum_{n =1}^{N-f} \log_{2} |\mathcal{S}_n| \geq \log_{2}|\mathcal{V}|.$$

	As a consequence, we have $\log_{2}|\mathcal{S}_{N-f}|\geq \frac{\log_{2} |\mathcal{V}|}{N-f}$. Therefore, we have $\displaystyle\max_{n \in \{1,2,\ldots,N\}} \log_{2}|\mathcal{S}_{n}| \geq  \log_{2}|\mathcal{S}_{N-f}| \geq \frac{\log_{2}|\mathcal{V}|}{N-f}.$ Furthermore, we have $\log_{2}|\mathcal{S}_{n}| \geq  \frac{\log_{2} |\mathcal{V}|}{N-f}$ for every $n \in \{N-f+1,\ldots,N\}$. This implies the following chain of relations.
	\begin{eqnarray*}
		\sum_{n =1}^{N} \log_{2} |\mathcal{S}_n| &\geq& \log_{2}|\mathcal{V}| + \sum_{n=N-f+1}^{N} \log_{2}|\mathcal{S}_{n}| \\
		&\geq&  \log_{2}|\mathcal{V}| + f \frac{\log_{2}|\mathcal{V}|}{N-f} = \frac{N\log_{2}|\mathcal{V}|}{N-f}
	\end{eqnarray*}
This completes the proof.
\end{proof}

We now prove Theorem \ref{thm:first}.


\subsection{Informal Proof Sketch for Theorem \ref{thm:first}} Intuitively, the above theorem can be understood as follows. Consider any subset $\mathcal{N} \subseteq \{1,2,\ldots,N\}$ where $|\mathcal{N}| = N-f$.
Consider an execution of the algorithm $A$ where servers $\{1,2,\ldots,N\} - \mathcal{N}$ fail {at the beginning of the execution}. After the servers fail, a writer writes value $v$ in the system and terminates. Because the algorithm is {regular}, any reader that begins after the termination of the write must recover the last written value $v$ from the $N-f$ servers in $\mathcal{N}$. {Since the state of server $i$ belongs to $\mathcal{S}_i$}, the total number of possible {configurations} of the {states of the }$N-f$ servers in $\mathcal{N}$ is $\prod_{i=1}^{N-f}|\mathcal{S}_i|$. Since the value $v$ can be any element of the set $\mathcal{V}$ and the reader must recover the value through messages exchanged with the servers, there must be a one-to-one mapping from the set of values to the set of server states. Therefore, we need $\prod_{n \in \mathcal{N}}|\mathcal{S}_n| \geq |\mathcal{V}|$, which implies the result of Theorem \ref{thm:first}. We provide a formal proof below. 
\subsection{Formal Proof of Theorem \ref{thm:first}}

\begin{proof}[{Proof of Theorem \ref{thm:first}}]
Consider any subset $\mathcal{N} \subseteq \{1,2,\ldots,N\}$ where $|\mathcal{N}| = N-f$.
We construct $|\mathcal{V}|$ executions of the algorithm. In particular, for every value $v$ in $\mathcal{V}$, we construct an execution $\alpha^{(v)}$ of the algorithm as follows. In $\alpha^{(v)}$, the $f$ servers in $\{1,2,\ldots,N\}-\mathcal{N}$ fail at the beginning of the execution. The servers in $\mathcal{N}$ do not fail in $\alpha^{(v)}.$ After the $f$ servers fail, a write operation with value $v$ begins and all components take turns in a fair manner until the write operation terminates. Since, in a fair execution of algorithm $A$ where {the} number of server failures is at most $f$, {any operation invoked at a non-failing client} eventually terminates, we can ensure that the execution can be extended until the write terminates. 
Let $\tilde{P}^{(v)}$ be some point after the termination of the write. At $\tilde{P}^{(v)}$, all the channels in the system act, delivering {all} their messages. Let $P^{(v)}$ {be some point} in $\alpha^{(v)}$ after the channels deliver their messages. {At} $P^{(v)},$ the write client {fails}. At some point after $P^{(v)}$, a read operation begins and all the components in the system take turns in a fair manner until the read terminates. {Because the read client does not fail, and because the number of server failures is $f$, the read operation terminates in $\alpha^{(v)}$}. The execution $\alpha^{(v)}$ ends after the completion of the read operation. 

For $j \in \mathcal{N}$ denote the state of server $j$ at point $P^{(v)}.$ We denote by $\vec{S}^{(v)},$ the tuple $(S_{j_1}^{(v)}, S_{j_2}^{(v)}, \ldots, S_{j_{N-f}}^{(v)})$ of server states, where $\mathcal{N}=\{j_1, j_2,\ldots, j_{N-f}\},$ and $j_1 < j_2 < \ldots j_{N-f}$. Note that $\vec{S}^{(v)}$ is an element from the set $\prod_{n \in \mathcal{N}} S_n.$  To complete the proof of the lemma, it suffices to show that for two distinct  values $v, v'$ in $\mathcal{V}$, we have {$\vec{S}^{(v)} \neq \vec{S}^{(v')}.$} This is because {$\vec{S}^{(v)} \neq \vec{S}^{(v')}$} implies that there are at least $|\mathcal{V}|$ elements in $\prod_{n \in \mathcal{N}} S_n,$ which implies the theorem statement owing to the following chain of relations:
\begin{align*}
& \prod_{n \in \mathcal{N}}|\mathcal{S}_n| \geq |\mathcal{V}| \\
 \Rightarrow & \sum_{n \in \mathcal{N}}\log_2 |\mathcal{S}_n| \geq \log_{2}|\mathcal{V}| 
\end{align*}
So, to complete the proof, {it is enough to show} that for distinct values $v \neq v',$ we have {$\vec{S}^{(v)} \neq \vec{S}^{(v')}.$} 

{Assume for contradiction} that there exist two different values $v\neq v'$ such that {$\vec{S}^{(v)} = \vec{S}^{(v')}$}. We create an execution $\beta$ of the algorithm where the operations are not {regular}, which would contradict the assumption that the algorithm is {regular} and complete the proof. The steps of $\beta$ are identical to the steps of execution $\alpha^{(v)}$ until the point $P^{(v)}$ in $\alpha^{(v)}$. {Consider the composite automaton that includes the servers, the readers, the channels amongst the servers, and the channels between the readers and the servers.} The state of every component {of this composite automaton} at point $P^{(v)}$ in $\alpha^{(v)}$ is the same as the state of the corresponding component at point $P^{(v')}$ in $\alpha^{(v')}$. This is because at both $P^{(v)}$ and $P^{(v')}$, all the channels are empty, the servers in $\{1,2,\ldots,N\} - \mathcal{N}$ have failed, and the state of any server from $\mathcal{N}$ is the corresponding element of $\vec{S}^{(v)},$ which is equal to $\vec{S}^{(v')}$. In $\beta,$ after the point $P^{(v)},$ all the components follow the steps of the execution $\alpha^{(v')}$. Clearly $\beta$ is an execution of the algorithm $A$. Because $P^{(v)}$ is after the termination of the write operation that wrote value $v$ in $\alpha^{(v)}$, and because $\beta$ is identical to $\alpha^{(v)}$ until the point $P^{(v)}$, we infer that a write operation wrote value $v$ in $\beta$. Because a read operation begins after $P^{(v')}$ and returns $v'$ in $\alpha^{(v')},$ and because every component follows the same steps in execution $\beta$ as in execution $\alpha^{(v')}$ after point $P^{(v')},$ we infer that a read operation begins in $\beta$ after the termination of the write and returns value $v'$. If $\beta$ is {regular}, this read operation is serialized after the write operation in $\beta$. Therefore, the read operation should return $v$. However it returns $v'$ which is not equal to $v$. Therefore $\beta$ is not {regular}, which contradicts the assumption that the algorithm is {regular}. Therefore, for any two different values $v\neq v',$ we have ${\vec{S}^{(v)} \neq \vec{S}^{(v')}}.$ This completes the proof.
\end{proof}

\end{document}